\documentclass[a4paper]{article}
\usepackage{geometry}
\usepackage{mathrsfs}
\usepackage{bm}
\usepackage{authblk}
\usepackage{xcolor}
\usepackage{comment}
\usepackage{upgreek}
\usepackage{amssymb}  
\usepackage{amsthm}
\usepackage{amsmath}
\usepackage{graphicx}

\newtheorem{theorem}{Theorem}[section]
\newtheorem*{remark}{Remark}

\newcommand{\field}{\mathbb}

\newcommand{\reals}{\field{R}}

\newcommand{\ie}{\textit{i.e.,}~}
\newcommand{\eg}{\textit{e.g.,}~}

\newcommand{\covder}[2]{\nabla_{#2}{#1}}

\newcommand{\mbs}[1]{\bm{#1}}
\newcommand{\dexp}{\mathrm{dexp}}
\renewcommand{\skew}{\mathrm{skew}}

\newcommand{\nodrill}{\mbs{\chi}}
\newcommand{\pmap}[1]{{#1}_\perp}
\newcommand{\dmap}[1]{{#1}_\parallel}
\newcommand{\rotation}{{\mbs{{\Lambda}}}}

\newcommand{\dd}[2]{\frac{\mathrm{d}#1}{\mathrm{d}#2}}
\newcommand{\pd}[2]{\frac{\partial #1}{\partial #2}}

\newcommand{\Order}[2]{\mathcal{O}(#1^{#2})}

\let\oldOmega=\Omega \renewcommand{\Omega}{\mathit{\oldOmega}}
\let\oldGamma=\Gamma \renewcommand{\Gamma}{\mathit{\oldGamma}}
\let\oldLambda=\Lambda \renewcommand{\Lambda}{\mathit{\oldLambda}}
\let\oldSigma=\Sigma \renewcommand{\Sigma}{\mathit{\oldSigma}}
\let\oldPi=\Pi \renewcommand{\Pi}{\mathit{\oldPi}}
\let\oldXi=\Xi \renewcommand{\Xi}{\mathit{\oldXi}}

\begin{document}


\providecommand{\keywords}[1]{\textbf{\textit{Keywords:}} #1}

\title{The rotating rigid body model based on a non-twisting frame}

\author{Cristian Guillermo Gebhardt, Ignacio Romero}

\date{}

\maketitle




\begin{abstract}
This work proposes and investigates a new model of the rotating rigid body based on the \emph{non-twisting} frame. Such a frame consists of three mutually orthogonal unit vectors whose rotation rate around one of the three axis remains zero at all times and thus, is represented by a nonholonomic restriction. Then, the corresponding Lagrange-D'Alembert equations are formulated by employing two descriptions, the first one relying on rotations and a splitting approach, and the second one relying on constrained directors. For vanishing external moments, we prove that the new model possesses conservation laws, \ie the kinetic energy and two nonholonomic momenta that substantially differ from the holonomic momenta preserved by the standard rigid body model. Additionally, we propose a new specialization of a class of energy-momentum integration schemes that exactly preserves the kinetic energy and the nonholonomic momenta replicating the continuous counterpart. Finally, we present numerical results that show the excellent conservation properties as well as the accuracy for the time-discretized governing equations.   
\end{abstract}


\section{Introduction}
\label{sec-intro}
The rigid body is a problem of classical mechanics that has been exhaustively studied (see,
\eg\cite{Goldstein:1980up,Arnold:1989wt}).  Its simplicity, as well as its relation with the
(nonlinear) rotation group, makes of it the ideal setting to study abstract concepts of geometric
mechanics, such as Poisson structures, reduction, symmetry, stability, etc. Additionally, many of
the ideas that can be analyzed in the context of the rigid body can later be exploited in the study
of nonlinear structural theories, such as rods and shells (as for example in
\cite{Simo:1988uc,Antman:1995wm,Mielke:1988hk}). As a result, geometrical integrators for the rigid
body \cite{Simo1991,Rom08} are at the root of more complex numerical integration schemes for models
that involve, in one way or another, rotations
\cite{Simo:1986ws,Sansour:1995fi,Jelenic:1998uu,Romero:2002wb,Betsch:2003vl,Romero:2017uv}.

More specifically, the role of the rotation group is key because it is usually chosen to be the
configuration space of the rigid body, when the latter has a fixed point. The rich Lie group
structure of this set is responsible for much of the geometric theory of the rigid body, but it is
not the only possible way to describe it. For example, the configuration of the a rigid body with a
fixed point can also be described with three mutually orthogonal unit vectors. While this alternative
description makes use of constraints, it has proven useful in the past for the construction of
conserving numerical methods for rigid
bodies, rods, and multibody systems \cite{Romero:2002uw,Betsch:2002wk,BetschLeyendecker06}.

In this article we explore a third route that can be followed to describe the kinematics of rigid
bodies. This avenue is based on introducing a \emph{non-twisting} or Bishop frame of motion
\cite{AmericanMathematica:6j0BCmRM}. This frame consists of three mutually orthogonal unit vectors
whose rotation rate around one of the three axis remains zero at all times. Such a frame has proven
useful to study the configuration of nonlinear Kirchhoff rods
\cite{Shi:1994ei,McMillen:2002iz,Audoly:2007ko,Romero2019}, but has not received enough attention in
the context of the rigid body.

Formulating the equations of motion for the rigid body in the \emph{non-twisting} frame demands a
construction that is different from the standard one. In particular, the definition of Bishop's
frame requires a constraint that is nonholonomic and does not admit a variational
statement. Additionally, conservation laws take in this setting a different form when compared with the usual ones and
whose identification is not straightforward, and the governing equations, \ie the
Lagrange-D'Alembert equations, are elegantly formulated by means of an splitting approach in terms
of the covariant derivative on the unit sphere. Some of the conservation laws that take place under
consideration of the \emph{non-twisitng} frame may substantially differ from other nonholonomic
cases that were investigated in the literature, \eg \cite{Betsch2006,Hedrih2019}. In a pure
holonomic context, some attempts to reformulate the dynamics on the unit sphere by means of advanced
concepts from the differential geometry are to be found in \cite{Lee2009,Lee2018}. However, the
anisotropy of the inertial properties has been completely disregarded. 

The rigid body equations with a nonholonomic constraint can be integrated in time with a
conserving scheme that preserves energy and the newly identified momenta, the so-called nonholonomic
momenta. This method, based on the idea of the average vector field
\cite{McLachlan:1999tm,Celledoni:2012ff} preserves remarkably these invariants of the motion,
exactly, resulting in accurate pictures of the rigid body dynamics. The specialization of approaches
based on more elaborated conservative/dissipative integration schemes like \cite{Gebhardt2019c} is
possible as well in this context, but falls outside the scope of the current work and therefore, not
addressed here.    

The rest of the article has the following structure. In Section~\ref{sec-geometry}, the fundamental
concepts from the differential geometry are presented and
discussed. Section~\ref{sec-standard-formualtion} presents two derivations of the equations of
motion for the standard rotating rigid body. The first set of equations is a \emph{split} version of
the well-known Euler equations that are presented within a novel geometrical framework that relies
on covariant derivatives. The second one, then,  produces a totally equivalent set of equations
and relies on three constrained directors. Such an approach possesses a very favorable mathematical
setting that will be exploited later on to derive an structure-preserving integration algorithm. In
Section~\ref{sec-nonholonomic-formulation}, the \emph{non-twisting} condition is enforced for
both fully equivalent formulations. Additionally, new conservation laws are identified in
the continuous setting. Section~\ref{sec-integration} starts with the energy-momentum integration
scheme for the director-based formulation and then it is modified for the nonholonomic case. The
conservation properties are identified analytically for the discrete setting, replicating their
continuous counterparts. In Section~\ref{sec-results}, numerical results that show the excellent
performance of the new energy-momentum algorithm are presented and discussed. Finally,
Section~\ref{sec-summary} closes this work with a summary.

\section{Relevant geometrical concepts}
\label{sec-geometry}
The governing equations of the rigid body are posed on nonlinear manifolds.
We briefly summarize the essential geometrical concepts required for a complete
description of the model (see, e.g., \cite{Marsden:1994wk,Lee2018} for more comprehensive
expositions).

\subsection{The unit sphere}
\label{subs-stwo} The unit sphere $S^2$ is a nonlinear, compact, two-dimensional manifold that
often appears in the configuration spaces of solid mechanics, be it the rigid body, rods, or
shells~\cite{Eisenberg1979,Simo:1989wu,Romero:2004uo,RoUrreCy:2014,Romero:2017uv}.  As an embedded
set on Euclidean space, it is defined as
\begin{equation}
  \label{eq-stwo}
  S^{2}
  :=
  \left\{
    \bm{d}\in\reals^{3}\mid\bm{d}\cdot\bm{d}=1
  \right\},
\end{equation}
where the dot product is the standard one in $\reals^3$. The tangent bundle of the unit 2-sphere
is also a manifold defined as
\begin{equation}
  \label{eq-tstwo}
  TS^2 :=
  \left\{
    (\bm{d},\bm{v}),\
    \bm{d}\in S^2,\
    \bm{v}\in\reals^3,\
    \bm{d}\cdot \bm{v} = 0
  \right\} .
\end{equation}
Alternatively, tangent vectors of $S^2$ to a point $\mbs{d}$ are those
defined by the expressions:
\begin{equation}
  \bm{v} = \bm{w}\times \bm{d}\, ,
  \quad
  \mathrm{with}
  \quad
  \bm{w}\cdot \bm{d}= 0\ ,
\end{equation}
where the product ``$\times$'' is the standard cross product in $\reals^3$.

In contrast with the space of rotations, to be studied in more detail later,
the unit sphere does not have a group structure, but instead it has that
of a Riemannian manifold. The connection of this set can be more easily
explained by considering it to be an embedded manifold in $\reals^3$.
As such, the covariant derivative of a smooth vector field $\bm{v}:S^2\to TS^2$ along a
vector field $\bm{w}:S^2\to TS^2$ is the vector field $\nabla_{\bm{w}}\bm{v}\in TS^2$
that evaluated at $\bm{d}$ is precisely the projection
of the derivative $D\bm{v}$ in the direction of $\bm{w}$ onto
the tangent plane to $\bm{d}$. Hence, denoting as $\mbs{I}$ the unit
second order tensor and  $\otimes$ the
outer product between vectors, both on $\reals^3$, 
this projection can be simply written as
\begin{equation}
  \label{eq-cov-projection}
  \nabla_{\bm{w}} \bm{v}
  :=
  (\bm{I} - \bm{d}\otimes \bm{d})
  \;
  D \bm{v} \cdot \bm{w}\ .
\end{equation}
When $\bm{d}:(a,b)\to S^2$ is a smooth one-parameter curve on the unit
sphere and $\bm{d}'$ its derivative, the covariant derivative
of a smooth vector field $\bm{v}:S^2\to TS^2$ in the direction of $\bm{d}'$
has an expression that follows from Eq.~(\ref{eq-cov-projection}), that is,
\begin{equation}
  \nabla_{\bm{d}'} \bm{v}
  =
  (\bm{I} - \bm{d}\otimes \bm{d})
  \;
  D \bm{v} \cdot \bm{d}'
  =
  (\bm{v}\circ \bm{d})'
  -
  \left(
    (\bm{v}\circ \bm{d})' \cdot \bm{d}
  \right)
  \bm{d}, 
\end{equation}
which, as before, is nothing but the projection of $(\bm{v}\circ \bm{d})'$ onto the tangent space
$T_{\bm{d}}S^2$, and the symbol ``$\circ$'' stands for composition. The covariant derivative allows
to compare two tangent vectors belonging to different tangent spaces. By means of the
parallel transport, one vector can be transported from its tangent space to the space
of the other one. Then, all the desired comparisons can be made over objects belonging
to the same space.

The exponential map $\exp:T_{\bm{d}_0}S^2 \to S^2$ is a surjective application
with a closed form expression given by the formula
\begin{equation}
\exp_{\bm{d}_0}(\bm{v}_0) = \cos(|\bm{v}_0|)\bm{d}_0+\sin(|\bm{v}_0|)\frac{\bm{v}_0}{|\bm{v}_0|}\, ,
\label{eq-exp-2-sphere}
\end{equation}
where $\bm{v}_0$ must be a tangent vector on $T_{\bm{d}_0}S^2$
and $|\cdot|$ denotes the Euclidean vector norm. The inverse
of the exponential function
is the logarithmic map $\log:S^2 \to T_{\bm{d}_0}S^2$, for which also there is a closed
form expression that reads
\begin{equation}
\log_{\bm{d}_0}(\bm{d}) = 
\arccos(\bm{d}_0\cdot\bm{d})\frac{(\bm{I} - \bm{d}_0\otimes \bm{d}_0)\bm{d}}{|(\bm{I} -
  \bm{d}_0\otimes \bm{d}_0)\bm{d}|} ,
\end{equation}
with $\bm{d}_0\neq-\bm{d}$. The geodesic $\bm{d}_\mathrm{G}(s)$ for $s\in[0,1]$ with
$\bm{d}_\mathrm{G}(0)=\bm{d}_0$ and $\bm{d}'_\mathrm{G}(0)=\bm{v}_0$ is
a great arch on the sphere obtained as the solution of the equation 
\begin{equation}
  \nabla_{\bm{d}'}{\bm{d}'} = 
  \bm{0}\, ,
\end{equation}
with the explicit form
\begin{equation}
\bm{d}_\mathrm{G}(s) = \cos(|\bm{v}_0| s)\bm{d}_0+\sin(|\bm{v}_0| s)\frac{\bm{v}_0}{|\bm{v}_0|}\,.
\end{equation}
The parallel transport of $\bm{w}_0\in T_{\bm{d}_0}S^2 \mapsto \bm{w} \in T_{\bm{d}}S^2$ along the geodesic $\bm{d}_\mathrm{G}$ is then given by
\begin{equation}
  \bm{w} =
  \left(
    \bm{I}-
    \frac{1}{\arccos^2(\bm{d}_0\cdot\bm{d}_\mathrm{G})}
    \left(
      \log_{\bm{d}_0}\bm{d}_\mathrm{G} + \log_{\bm{d}_\mathrm{G}} \bm{d}_0
    \right)
    \otimes
    \log_{\bm{d}_0}\bm{d}_{\mathrm{G}}
  \right) \bm{w}_0\,, 
\label{eq-parallel-transport-2-sphere}
\end{equation}
and verifies
\begin{equation}
\nabla_{\bm{w}}{\bm{d}'} = (\bm{I} - \bm{d} \otimes \bm{d})\bm{w}' = \bm{0}\,.
\end{equation}
More details about the expressions \eqref{eq-exp-2-sphere} to \eqref{eq-parallel-transport-2-sphere}
can be found in \cite{Hosseini2017,Bergmann2018}.

Given the definitions \eqref{eq-stwo} and \eqref{eq-tstwo} of the unit sphere
and its tangent bundle, we recognize that there exists an isomorphism
\begin{equation}
  \mathbb{R}^3 \cong T_{\mbs{d}}S^2 \oplus \mathrm{span}(\mbs{d}),
  \label{eq-s2-isomorphism}
\end{equation}
for any $\mbs{d}\in S^2$. Given now two directors
$\mbs{d},\tilde{\mbs{d}}$, we say that a second order tensor $\mbs{T}:\reals^3\to\reals^3$
\emph{splits from $\mbs{d}$ to $\tilde{\mbs{d}}$} if it can be written in the form
\begin{equation}
  \mbs{T}= \pmap{\mbs{T}} + \dmap{\mbs{T}}\:,
  \label{eq-t-split}
\end{equation}
where $\pmap{\mbs{T}}$ is a bijection from $T_{\mbs{d}} S^2$
to $T_{\tilde{\mbs{d}}}S^2$ with $\ker(\pmap{\mbs{T}}) = \hbox{span}(\mbs{d})$,
and $\dmap{\mbs{T}}$ is a bijection from $\hbox{span}(\mbs{d})$ to $\hbox{span}(\tilde{\mbs{d}})$
with $\ker(\dmap{\mbs{T}}) \equiv T_{\mbs{d}}S^2$. The
split~\eqref{eq-t-split} depends on the pair~$\mbs{d},\tilde{\mbs{d}}$
but it is not indicated explicitly in order to simplify the notation.

\subsection{The special orthogonal group}
\label{subs-rotation}
Classical descriptions of rigid body kinematics are invariably based on
the definition of their configuration space as the set of proper orthogonal
second order tensors, that is, the special orthogonal group, defined as
\begin{equation}
  SO(3) :=
  \left\{
    \bm{\Lambda}\in \reals^{3\times3}, \
    \bm{\Lambda}^T \bm{\Lambda} = \bm{I}, \
    \det{\bm{\Lambda}} = +1
    \right\} .
  \label{eq-sothree}
\end{equation}
This smooth manifold has a group-like structure when considered with the tensor
multiplication operation, thus it is a Lie group. Its associated Lie algebra
is the linear set
\begin{equation}
  so(3) :=
  \left\{
    \hat{\bm{w}}\in\reals^{3\times3}, \
    \hat{\bm{w}} = - \hat{\bm{w}}^T
  \right\}.
  \label{eq-algebra}
\end{equation}
Later, it will be convenient to exploit the isomorphism that
exists between three-dimensional real vectors and $so(3)$.
To see this, consider the map $\hat{(\cdot)}:\reals^3\to
so(3)$ such that for all $\bm{w},\bm{a}\in\reals^3$, the tensor $\hat{\bm{w}}\in so(3)$ satisfies
$\hat{\bm{w}}\bm{a} = \bm{w}\times \bm{a}$. Here, $\mbs{w}$ is referred to as the axial vector
of $\hat{\mbs{w}}$, which is a skew-symmetric tensor, and we also write $\skew[\bm{w}] =
\hat{\bm{w}}$. Being a Lie group, the space of rotations has an exponential
map $\exp:so(3)\to SO(3)$ whose closed form expression is Rodrigues' formula
\begin{equation}
  \exp[\hat{\bm{\theta}}]
  :=
  \bm{I} +  \frac{\sin\theta}{\theta} \hat{\bm{\theta}}
  +
  \frac{1}{2} \frac{\sin^2(\theta/2)}{(\theta/2)^2} \hat{\bm{\theta}}^2\ ,
  \label{eq-rodrigues}
\end{equation}
with $\bm{\theta}\in\reals^3$,  $\theta =|\bm{\theta}|$. The linearization
of the exponential map is simplified by introducing
the map $\dexp: so(3)\to\reals^{3\times 3}$ that
satisfies
\begin{equation}
  \dd{}{\epsilon}
  \exp[\hat{\bm{\theta}}(\epsilon)]
  =
  \skew\left[\dexp[\hat{\bm{\theta}}(\epsilon)] \dd{}{\epsilon}\bm{\theta}(\epsilon)\right] 
  \exp[\hat{\bm{\theta}}(\epsilon)]
  \label{eq-dexp}
\end{equation}
for every $\bm{\theta}:\reals\to\reals^3$.
A closed-form expression for this map, as well as more details regarding
the numerical treatment of rotations can be found elsewhere \cite{Hairer:2002vg,Rom08,Romero:2017uv}.

\subsection{The notion of twist and the \emph{non-twisting} frame}
\label{subs-non-twisting-frame}

Let $\bm{d}_3:\left[0,T\right]\rightarrow S^2$ indicate a curve of directors parameterized by time
$t\in[0,T]$.  Now, let us consider two other curves $\bm{d}_1,\bm{d}_2:[0,T]\to S^2$ such that
$\{\bm{d}_1(t), \bm{d}_2(t), \bm{d}_3(t)\}$ are mutually orthogonal for all $t\in[0,T]$.  We say
that the triad $\{\bm{d}_1,\bm{d}_2,\bm{d}_3\}$ moves without twist if
\begin{equation}
  \dot{\bm{d}}_1\cdot\bm{d}_2 = \dot{\bm{d}}_2\cdot\bm{d}_1  = 0\ ,
  \label{eq-non-twisting}
\end{equation}
where the overdot refers to the derivative with respect to time.
Given the initial value of the triad
$\{\mbs{d}_1(0), \mbs{d}_2(0), \mbs{d}_3(0) \} = \{ \mbs{D}_1,\mbs{D}_2, \mbs{D}_3\}$,
there is a map $\mbs{\chi}:[0,T]\to SO(3)$ transforming it to
the frame $\{ \mbs{d}_1, \mbs{d}_2, \mbs{d}_3 \}$ that evolves without twist, and  whose
closed form is 
\begin{equation}
\bm{\chi} = \bm{d}_1\otimes\bm{D}_1+\bm{d}_2\otimes\bm{D}_2+\bm{d}_3\otimes\bm{D}_3\, .
\end{equation}
%
%
The \emph{non-twisting} frame has Darboux vector
\begin{equation}
  \bm{w}_{\bm{\chi}} = \bm{d}_3\times\dot{\bm{d}}_3\:,
\end{equation}
and it is related to parallel transport in the sphere.
To see this relation, consider again the same non-twisting frame as before. 
Then, we recall that a vector field $\bm{v} \in T_{\bm{d}_3}S^2$ is said to be parallel-transported
along $\bm{d}_3$ if and only if $\nabla_{\bm{v}}{\dot{\bm{d}}_3}=\bm{0}$. An consequence of this is that
\begin{equation}
\dot{\bm{d}}_{3}\cdot\dot{\bm{d}}_{1} 
= \dot{\bm{d}}_{3}\cdot(\bm{w}_{\bm{\chi}}\times\bm{d}_{1})
= 0 ,
\end{equation}
and similarly
\begin{equation}
\dot{\bm{d}}_{3}\cdot\dot{\bm{d}}_{2} = \dot{\bm{d}}_{3}\cdot(\bm{w}_{\bm{\chi}}\times\bm{d}_{2}) = 0\, .
\end{equation}
In addition, we have that
\begin{equation}
\begin{split}
\dot{\bm{d}}_{1}\cdot\dot{\bm{d}}_{2} 
= (\bm{w}_{\bm{\chi}}\times\bm{d}_{1})\cdot(\bm{w}_{\bm{\chi}}\times\bm{d}_{2})
= (\bm{d}_1\cdot\dot{\bm{d}}_3)(\bm{d}_2\cdot\dot{\bm{d}}_3)\, ,
\end{split}
\end{equation}
which is merely the product of the angular velocity components and can be interpreted
as a scalar curvature.

To define precisely the concept of twist, let us consider the rotation
$\exp[\psi\:\hat{\bm{d}}_3]$, with $\psi:[0,T]\to S^1$ (the unit 1-sphere), and the rotated triad
$\{\bm{d}_1^\psi,\bm{d}_2^\psi,\bm{d}_3\} = \exp[\psi\:\hat{\bm{d}}_3]
\{\bm{d}_1,\bm{d}_2,\bm{d}_3\}$.
Then, 
\begin{equation}
  \bm{d}_1^{\psi} = \cos(\psi)\bm{d}_1+\sin(\psi)\bm{d}_2
  \qquad \mathrm{and}\qquad
  \bm{d}_2^{\psi} =-\sin(\psi)\bm{d}_1+\cos(\psi)\bm{d}_1 ,
\end{equation}
and the Darboux vector of the rotated triad is
\begin{equation}
	\bm{w}_{\exp[\psi\:\hat{\bm{d}}_3]\bm{\chi}} = \bm{d}_3\times\dot{\bm{d}}_3+\dot{\psi}\bm{d}_3,
\end{equation}
or, equivalently,
\begin{equation}
  \bm{w}_{\exp[\psi\:\hat{\bm{d}}_3]\bm{\chi}} = -(\bm{d}_2^{\psi}\cdot\dot{\bm{d}}_3)\bm{d}_1^{\psi}
  +(\bm{d}_1^{\psi}\cdot\dot{\bm{d}}_3)\bm{d}_2^{\psi}+\dot{\psi}\bm{d}_3\, .
\label{eq-w_twist}
\end{equation}
In this last expression, we identify the twist rate $\dot{\psi}$ and
the twist angle $\psi$, respectively, as the rotation velocity
component on the $\mbs{d}_3$ direction and the rotation angle about this
same vector (for further details, see
\cite{AmericanMathematica:6j0BCmRM,Langer:1996wj}). The previous calculations show that
the frame $\{\bm{d}_1,\bm{d}_2,\bm{d}_3\}$ --- known
as the natural or Bishop frame in the context of one-parameter curves embedded in the ambient space ---
is the unique one obtained by transporting $\{\bm{D}_1, \bm{D}_2,\bm{D}_3\}$
\emph{without twist}. In this context, a summary of recent advances and
open problems are presented for instance in \cite{Farouki2016}. 
\section{Standard rotating rigid body}
\label{sec-standard-formualtion}
In this section, we review the classical rotating rigid body model, which we take as the starting
point for our developments.  We present this model, however, in an unusual fashion. In it, the
governing equations of the body appear in \textit{split form}.  This refers to the fact that, for a
given director $\bm{d}_3$, the dynamics of the body that takes place in the cotangent space
$T^*_{\bm{d}_3}S^2$ is separated from that one corresponding to the reciprocal normal
space $N^*_{\bm{d}_3}S^2 \equiv \textrm{span}(\bm{d}_3)$.
In order to do this, we have employed the identifications
\begin{equation}
  \reals^3 \cong T_{\bm{d}_3}S^2 \oplus \textrm{span}(\bm{d}_3) \cong T^*_{\bm{d}_3}S^2 \oplus
  \textrm{span}(\bm{d}_3)\, .
  \label{eq-r3-isomorphism}
\end{equation}

\subsection{Kinematic description}
\label{subs-general}
As customary, a rotating rigid body is defined to be a three-dimensional non-deformable body. The
state of such a body, when one of its points is fixed,
can be described by a rotating frame whose orientation is given by a rotation
tensor. Thus, the configuration manifold is $Q\equiv SO(3)$.

Let us now study the motion of a rotating rigid body, that is, a time-parameterized curve in
configuration space $\mbs{\Lambda}: [0,T] \to Q$. The generalized velocity of the rotating rigid
body belongs, for every $t\in[0,T]$, to the tangent bundle
\begin{equation}
  TQ := 
  \left\{(\bm{\Lambda},\dot{\bm{\Lambda}}),
    \mbs{\Lambda}\in SO(3),
    \mbs{\Lambda}^T \dot{\mbs{\Lambda}} \in so(3)
  \right\} .
  \label{eq-tq}
\end{equation}
The time derivative of the rotation tensor can be written as
\begin{equation}
  \dot{\bm{\Lambda}} = \hat{\mbs{w}} \bm{\Lambda} = \bm{\Lambda} \widehat{ \mbs{W}}\ ,
  \label{eq-omegas}
\end{equation}
where $\mbs{w}$ and $\mbs{W}$ are the spatial and convected angular
velocities, respectively.

Let $\{\mbs{E}_1,\mbs{E}_2, \mbs{E}_3\}$ be
a fixed basis of the ambient space. Then, if $\mbs{d}_i=\mbs{\Lambda} \mbs{E}_i$, with $i=1,2,3$
we can use Eq.~\eqref{eq-r3-isomorphism} to split the rotation vectors as in
\begin{equation}
  \mbs{w} = \mbs{w}_{\perp} + w_\parallel \bm{d}_3\ ,
  \qquad
  \mbs{W} = \mbs{W}_\perp + W_\parallel \bm{d}_3\ .
  \label{eq-omegas-split}
\end{equation}
Then, using the relations~\eqref{eq-omegas}, we identify 
\begin{equation}
  \mbs{w}_{\perp} = \mbs{d}_3\times \dot{\mbs{d}_3}\ ,
  \qquad
  \mbs{W}_{\perp} = \bm{\Lambda}^T \mbs{w}_{\perp}\ .
  \qquad
  \label{eq-omegas-perp}
\end{equation}

\subsection{Kinetic energy and angular momentum}
Let us now select a fixed Cartesian basis of $\mathbb{R}^3$ denoted as $\{ \mbs{D}_i \}_{i=1}^3$,
where the third vector coincides with one of the principal directions of the convected inertia
tensor $\mbs{J}:\mathbb{R}^3\to\mathbb{R}^3$ of the body,
a symmetric, second order, positive definite tensor. Thus, this tensor splits from $\mbs{D}_3$
to $\mbs{D}_3$ and we write
\begin{equation}
  \bm{J} = \bm{J}_\perp + J_\parallel\, \bm{D}_3\otimes \bm{D}_3\ ,
  \label{eq-irho-decomp}
\end{equation}
where $\bm{J}_\perp$ maps bijectively $\mathrm{span}(\bm{D}_1,\bm{D}_2)$ onto
itself and satisfies $\bm{J}_\perp \bm{D}_3 =
\bm{0}$.

The kinetic energy of a rigid body with a fixed point is defined as
the quadratic form
\begin{equation}
  K :=
  \frac{1}{2} \mbs{W} \cdot \bm{J} \mbs{W}
  =
  \frac{1}{2} \mbs{w} \cdot \mbs{j} \mbs{w}
  \ ,
  \label{eq-kinetic-density}
\end{equation}
where $\bm{j}$ is the spatial inertia tensor, the push-forward of the
convected inertia, and defined as
\begin{equation}
  \bm{j} :=
  \bm{\Lambda} \bm{J} \bm{\Lambda}^T\ .
  \label{eq-inertia-inertia}
\end{equation}
Let $\mbs{d}_3=\mbs{\Lambda}\mbs{D}_3$. Given the relationship between the convected and spatial inertia,
it follows that the latter also splits, this time from
$\mbs{d}_3$ to $\mbs{d}_3$, and thus
\begin{equation}
  \bm{j} = \bm{j}_\perp + j_\parallel\, \bm{d}_3\otimes \bm{d}_3\ ,
  \label{eq-irho-decomp-spatial}
\end{equation}
where now $\bm{j}_\perp$ maps bijectively $\mathrm{span}(\bm{d}_1,\bm{d}_2)$ onto
itself and satisfies $\bm{j}_\perp \bm{d}_3 =
\bm{0}$.

As a consequence of the structure of the inertia tensor, the kinetic energy of a rotating
rigid body can be written in either of the following equivalent ways:
\begin{equation}
  \frac{1}{2} \mbs{w} \cdot \bm{j} \mbs{w}
  =
  \frac{1}{2} \mbs{W} \cdot \bm{J} \mbs{W}
  =
  \frac{1}{2} \mbs{W}_\perp \cdot \bm{J}_\perp \mbs{W}_\perp
  +
  \frac{1}{2} J_\parallel\, W_\parallel^2
  =
  \frac{1}{2} \mbs{w}_\perp \cdot \bm{j}_\perp \mbs{w}_\perp
  +
  \frac{1}{2} j_\parallel\, w_\parallel^2 .
  \label{eq-kinetic-density-convected}
\end{equation}
The angular momentum of the rotating rigid body is conjugate
to the angular velocity as in
\begin{equation}
  \qquad
  \mbs{\pi}
  :=
  \pd{K}{\mbs{w}} = \mbs{j} \mbs{w}\:,
  \label{eq-momenta}
\end{equation}
and we note that we can introduce a convected version of
the momentum $\mbs{\pi}$ by pulling it back with the rotating tensor
and defining
\begin{equation}
  \mbs{\Pi} := \bm{\Lambda}^T \mbs{\pi} = \pd{K}{\mbs{W}}  \ .
  \label{eq-pi-convected}
\end{equation}
Due to the particular structure of the inertia,
the momentum can also be split, as before, as in
\begin{equation}
  \mbs{\pi} = \mbs{\pi}_\perp + \pi_\parallel \bm{d}_3\ ,
  \qquad
  \mbs{\Pi} = \mbs{\Pi}_\perp + \Pi_\parallel \bm{d}_3,
  \label{eq-pi-split}
\end{equation}
with
\begin{equation}
  \mbs{\pi}_\perp = \mbs{j}_\perp \mbs{w}_\perp\ ,
  \qquad
  \mbs{\Pi}_\perp = \mbs{J}_\perp \mbs{W}_\perp ,
  \qquad
  \pi_\parallel = \Pi_\parallel = j_\parallel w_\parallel = J_\parallel W_\parallel \ .
  \label{eq-pi-split-components}
\end{equation}

\subsection{Variations of the motion rates}
\label{subs-k-variations}
The governing equations of the rigid body will be obtained using Hamilton's
principle of stationary action, using calculus of variations.
We gather next some results that will prove
necessary for the computation of the functional derivatives and, later,
for the linearization of the model.

To introduce these concepts, let us consider a curve of configurations
$\bm{\Lambda}_\iota (t)$ parametrized by the scalar $\iota$ and given by
\begin{equation}
\bm{\Lambda}_\iota(t)
=
\exp[\iota\, \widehat{\delta\bm{\theta}}(t)] \bm{\Lambda}(t)\, ,
\label{eq-one-parameter}
\end{equation}
where $\delta\bm{\theta} :[0,T]\to\reals^3$ represents all arbitrary variations
that satisfy
\begin{equation}
\delta \bm{\theta}(0) = \delta \bm{\theta}(T) = \bm{0} \ .
\label{eq-variations}
\end{equation}
The curve $\bm{\Lambda}_\iota$ passes through the configuration
$\bm{\Lambda}$ when $\iota=0$ and has tangent at this point
\begin{equation}
\left.\pd{}{t}\right|_{\iota=0} \bm{\Lambda}_\iota
=
\widehat{\delta\bm{\theta}} \bm{\Lambda}\, .
\label{eq-tangent-variation}
\end{equation}
For future reference, let us calculate the variation of the derivative $\dot{\bm{\Lambda}}$.
To do so, let us first define the temporal derivative of the perturbed rotation, that is,
\begin{equation}
\begin{split}    
\pd{}{t} \bm{\Lambda}_\iota
&=
\pd{}{t}\exp[\iota \widehat{\delta\bm{\theta}}] \bm{\Lambda}
=
\skew\left[ \dexp[\iota \widehat{\delta \bm{\theta}}] \iota \delta \dot{\bm{\theta}}\right]
\bm{\Lambda} + \exp[\iota \widehat{\delta \bm{\theta}}] \dot{\bm{\Lambda}} \ .
\end{split}
\label{eq-d-lambdaprime}
\end{equation}
Then, the variation of $\dot{\bm{\Lambda}}$ is just
\begin{equation}
\delta ( \dot{\bm{\Lambda}} )
=
\left.\pd{}{\iota}\right|_{\iota=0}
\pd{}{t} \bm{\Lambda}_\iota
=
\widehat{\delta \dot{\bm{\theta}}} \bm{\Lambda} + \widehat{\delta \bm{\theta}} \dot{\bm{\Lambda}}\ .
\label{eq-var-lambdaprime}
\end{equation}
With the previous results at hand, we can now proceed to calculate the variations of
the convected angular velocities, as summarized in the following theorem.

\begin{theorem}
	The variations of the convected angular velocities $(\bm{W}_\perp, W_\parallel)$
	are
	\begin{subequations}
	\begin{align}
	\delta \bm{W}_\perp
	&=
	\rotation^T
	(\bm{I}-\bm{d}_3 \otimes \bm{d}_3) \delta \dot{\bm{\theta}}\, ,
	\\
	\delta W_\parallel &= \bm{d}_3 \cdot \delta \dot{\bm{\theta}}\, .
	\end{align}
	\label{eq-variation-convected}
	\end{subequations}
\end{theorem}
\begin{proof}
	The convected angular velocities of the one-parameter curve of configurations
	$\bm{\Lambda}_\iota$ are
	\begin{equation}
	\bm{W}_{\perp,\iota} = \bm{D}_3 \times ( \bm{\Lambda}^T_{\iota} \dot{\bm{d}}_{3,\iota})
	\quad\textrm{and}\quad
	W_{\parallel,_\iota} =  \bm{d}_{2,\iota} \cdot  \dot{\bm{d}}_{1,\iota}\, ,
	\label{eq-str-proof-0}
	\end{equation}
	where $\bm{d}_{i,\iota} = \bm{\Lambda}_\iota \bm{D}_i$. 
	The variation of the angular velocity perpendicular to $\bm{D}_3$ is obtained from its
	definition employing some algebraic manipulations and
	expression~\eqref{eq-var-lambdaprime} as follows:
	\begin{equation}
	\begin{split}
	\delta \bm{W}_\perp
	&=
	\left. \pd{}{\iota} \right|_{\iota=0}
	\left(
	\bm{D}_3\times\left(\bm{\Lambda}^T_\iota \dot{\bm{d}}_{3,\iota}\right) 
	\right)
	\\
	&=
	\bm{D}_3 \times
	\left(
	\delta \bm{\Lambda}^T \dot{\bm{\Lambda}} \bm{D}_3 + \bm{\Lambda}^T
	\delta \dot{\bm{\Lambda}} \bm{D}_3
	\right)
	\\
	&=
	\bm{\Lambda}^T
	\left(
	\bm{d}_3 \times
	\left( \widehat{\delta \dot{\bm{\theta}}}\times \bm{d}_3 \right)
	\right)
	\\
	&=
	\bm{\Lambda}^T
	\left(
	\delta \dot{\bm{\theta}} - \left(\delta \dot{\bm{\theta}} \cdot \bm{d}_3\right) \bm{d}_3
	\right)
	\\
	&=
	\bm{\Lambda}^T
	(\mbs{I}-\bm{d}_3 \otimes \bm{d}_3) \delta \dot{\bm{\theta}}    
	\ .
	\end{split}
	\label{eq-str-proof-2}
	\end{equation}
	The variation of the angular velocity parallel to $\bm{D}_3$  follows similar steps:
	\begin{equation}
	\begin{split}
	\delta W_\parallel
	&=
	\left. \pd{}{\iota} \right|_{\iota=0}
	\left( \bm{d}_{2,\iota}\cdot \dot{\bm{d}}_{1,\iota} \right)
	\\
	&=
	\widehat{\delta \bm{\theta}} \bm{\Lambda} \bm{D}_2\cdot \dot{\bm{\Lambda}} \bm{D}_1
	+
	\bm{\Lambda} \bm{D}_2 \cdot
	\left(
	\widehat{\delta\dot{\bm{\theta}}}\bm{\Lambda} + \widehat{\delta\bm{\theta}}\dot{\bm{\Lambda}}
	\right)
	\bm{D}_1
	\\
	&=
	\bm{d}_1\times \bm{d}_2 \cdot \delta \dot{\bm{\theta}}
	\\
	&=
	\bm{d}_3 \cdot \delta \dot{\bm{\theta}}\ .    
	\end{split}
	\label{eq-str-proof-3}
	\end{equation}
\end{proof}

\subsection{Governing equations and invariants}
Here, we derive the governing equations of the rotating rigid body model
and the concomitant conservation laws.
Hamilton's principle of stationary action states that the governing equations
are the Euler-Lagrange equations of the action functional
\begin{equation}
	S = \int^T_0 K\, \textrm{d}t\, ,
\end{equation}
with unknown fields $(\bm{\Lambda}, \dot{\bm{\Lambda}})\in TQ$. 
\begin{theorem}
  The equations of motion, \ie the Euler-Lagrange equations,
  for the standard rotating rigid body model in split form are:

\begin{subequations}
	\begin{align}
 		\covder{\bm{\pi}_{\perp}}{\dot{\bm{d}}_3}+\pi_{\parallel}\dot{\bm{d}}_3 & = \bm{0}\, ,\\
 		\dot{\pi}_{\parallel}+\bm{\pi}_\perp\cdot\dot{\bm{d}}_3 &= 0\, .
	\end{align}
	\label{eq-standard}
\end{subequations}

\noindent The pertaining initial conditions are:

\begin{equation}
\bm{\Lambda}(0) = \bar{\bm{\Lambda}}\, ,\quad \bm{w}_\perp(0) = \bar{\bm{w}}_\perp\, ,\quad w_\parallel(0) = \bar{w}_\parallel\, .
\end{equation}
\end{theorem}
\begin{proof}
The theorem follows from the systematic calculation of $\delta S$, the variation of the action, based on the variation of the convected angular velocities of Eq. \eqref{eq-variation-convected}, thus we omit a detailed derivation. Eq. \eqref{eq-standard}, which is presented in its split form, is equivalent to Euler's equations which state that the spatial angular momentum is preserved, \ie $\dot{\bm{\pi}}=\bm{0}$. This is easily proven as follows
\begin{equation}
	\begin{split}
	\bm{0}
	&=
	\dot{\bm{\pi}}\\
	&= 
	(\mbs{I}-\bm{d}_3 \otimes \bm{d}_3)\dot{\bm{\pi}}+(\bm{d}_3\cdot\dot{\bm{\pi}})\bm{d}_3\\
	& = 
	(\mbs{I}-\bm{d}_3 \otimes \bm{d}_3)(\dot{\bm{\pi}}_\perp+\dot{\pi}_\parallel\bm{d}_3+\pi_\parallel\dot{\bm{d}}_3)+\bm{d}_3\cdot(\dot{\bm{\pi}}_\perp+\dot{\pi}_\parallel\bm{d}_3+\pi_\parallel\dot{\bm{d}}_3)\bm{d}_3\\
	& = \underbrace{\covder{\bm{\pi}_{\perp}}{\dot{\bm{d}}_3}+\pi_{\parallel}\dot{\bm{d}}_3}_{\in\: T^{*}_{\bm{d}_3}S^2}+\underbrace{(\dot{\pi}_\parallel+\dot{\bm{\pi}}_\perp\cdot\bm{d}_3)\bm{d}_3}_{\in\:\mathrm{span}(\bm{d}_3)}\, .
	\end{split}
\end{equation}
\end{proof}

\begin{theorem}
	The conservation laws of the rotating rigid body are:
	\begin{subequations}
	\begin{align}
	K        &= \frac{1}{2}\bm{W}\cdot\bm{J}\bm{W} = \frac{1}{2}\bm{w}\cdot\bm{j}\bm{w} = \text{const.}\, ,\\
    \bm{\pi} &= \bm{j}\bm{w} = \bm{\Lambda}\bm{J}\bm{W} = \text{const.}\, .         
	\end{align}
\end{subequations}	
\end{theorem}
\begin{proof}
This is an standard result and thus, we omit further details.
\end{proof}

\begin{remark}
  To include external moments acting on the standard rotating rigid body
  it is necessary to calculate the associated virtual work as follows
  \begin{equation}
    \delta W = \delta \bm{\theta}_\perp\cdot\bm{m}^{\mathrm{ext}}_\perp+\delta\theta_\parallel m^{\mathrm{ext}}_\parallel\, ,
  \end{equation}
  and add this contribution to the variation of the action.
\end{remark}

\subsection{Model equations based on directors}
Here, we present an alternative set of governing equations for the rotating rigid body model that will be used later to formulate a structure preserving algorithm. For this purpose, let us define the following configuration space
\begin{equation}
	Q:=\{\bm{q}=(\bm{d}_1,\bm{d}_2,\bm{d}_3)\in S^2\times S^2\times S^2\mid \bm{d}_i\cdot\bm{d}_j=0, i\neq j\}\cong SO(3)\,,
\end{equation}
whose tangent space at the point $\bm{q}$ is given by
\begin{equation}
T_{\bm{q}}Q:=\{\dot{\bm{q}}=(\dot{\bm{d}}_1,\dot{\bm{d}}_2,\dot{\bm{d}}_3)\mid \dot{\bm{d}}_i=\bm{\omega}\times\bm{d}_i,\bm{\omega}\in\reals^3\}\,.
\end{equation}

Now, we start by defining the rigid body as the bounded set ${\cal B}_0\subset\mathbb{R}^3$
of points
\begin{equation}
  \mbs{X} = \theta^1 \mbs{D}_1 + \theta^2 \mbs{D}_2 + \theta^3 \mbs{D}_3
  \label{eq-reference-conf}
\end{equation}
where $(\theta^1,\theta^2,\theta^3)$ are the material coordinates of the
point and $\{ \mbs{D}_i\}_{i=1}^3$ are three orthogonal directors, with the
third one oriented in the direction of the principal axis of inertia
and such that $\mbs{D}_3=\mbs{D}_1\times \mbs{D}_2$. 
The position of the point $\mbs{X}$ at time $t\in[0,T]$ is denoted as $\mbs{x}(t)\in\reals^3$
and given by
\begin{equation}
  \mbs{x}(t)
  =
  \bm{\varphi}(\theta^1,\theta^2,\theta^3;t)
  =
  \theta^{1}\bm{d}_1(t)+
  \theta^{2}\bm{d}_2(t)+
  \theta^{3}\bm{d}_3(t)
\label{eq_position_vector_rigidbody}
\end{equation}
with $(\bm{d}_1,\bm{d}_2,\bm{d}_3) = \bm{q} \in Q$, for all $t$. On this basis, there must be a rotation tensor $\mbs{\Lambda}(t)= \mbs{d}_i(t)\otimes \mbs{D}_i$, where we have employed the sum convention for repeated indices, such
that $\mbs{d}_i(t) = \mbs{\Lambda}(t) \mbs{D}_i$. The material velocity of the particle $\mbs{X}$ is the vector $\dot{\bm{x}}(t)\in \reals^3$
that can be written as
\begin{equation}
  \dot{\bm{x}}(t)
  =
  \dot{\bm{\varphi}}(\theta^1,\theta^2,\theta^3;t)
  =
  \theta^{1}\dot{\bm{d}}_1(t)+
  \theta^{2}\dot{\bm{d}}_2(t)+
  \theta^{3}\dot{\bm{d}}_3(t)
  \label{eq_velocity_vector_rigidbody}
\end{equation}
with $(\dot{\bm{d}}_1,\dot{\bm{d}}_2,\dot{\bm{d}}_3)=\dot{\bm{q}}\in T_{\bm{q}}Q$ representing three director velocity vectors.

To construct the dynamic equations of the model, assume the body
${\cal B}_0$ has a density $\rho_0$ per unit of reference volume and
hence its total kinetic energy, or Lagrangian, can be formulated as
\begin{equation}
  K = \int_{{\cal B}_0} \frac{\rho_0}{2}|\dot{\mbs{x}}|^2
  \,\mathrm{d} {\cal B}_0\ .
  \label{eq-director-K}
\end{equation}
To employ Hamilton's principle of stationary action, but
restricting the body directors to remain orthonormal at
all time, we define the constrained action
\begin{equation}
  S
  =\int_0^T
  \left(
    K
    - \mbs{h}(\mbs{d}_1,\mbs{d}_2,\mbs{d}_3) \cdot \mbs{\lambda}
\right) 
  \,\mathrm{d} t\ ,
  \label{eq-constrained-L}
\end{equation}
where $K$ is given by Eq.~\eqref{eq-director-K}, $\mbs{\lambda}\in\reals^3$
is a vector of Lagrange multipliers, and $\mbs{h}$ is of the form
\begin{equation}
  \bm{h}(\mbs{d}_1,\mbs{d}_2,\mbs{d}_3)
  =
  \begin{pmatrix}
    \bm{d}_2\cdot\bm{d}_3\\
    \bm{d}_1\cdot\bm{d}_3\\
    \bm{d}_1\cdot\bm{d}_2
  \end{pmatrix},
  \label{eq_internal_constraint}
\end{equation}
such that $\mbs{h}(\mbs{d}_1,\mbs{d}_2,\mbs{d}_3)=\mbs{0}$ expresses
the directors' orthonormality.

\begin{theorem}
The alternative equations of motion, \ie the Euler-Lagrange equations,
for the standard rotating rigid body model are:

\begin{subequations}
\begin{align}
\dot{\bm{\pi}}^1(\bm{d}_1,\bm{d}_2,\bm{d}_3)+\bm{H}_1(\bm{d}_1,\bm{d}_2,\bm{d}_3)^{T}\bm{\lambda}&=\bm{0}\,,\\
\dot{\bm{\pi}}^2(\bm{d}_1,\bm{d}_2,\bm{d}_3)+\bm{H}_2(\bm{d}_1,\bm{d}_2,\bm{d}_3)^{T}\bm{\lambda}&=\bm{0}\,,\\
\dot{\bm{\pi}}^3(\bm{d}_1,\bm{d}_2,\bm{d}_3)+\bm{H}_3(\bm{d}_1,\bm{d}_2,\bm{d}_3)^{T}\bm{\lambda}&=\bm{0}\,,\\
\mbs{h}(\mbs{d}_1,\mbs{d}_2,\mbs{d}_3) &= \bm{0}\,.
\end{align}
\label{eq_weak_form_rigid_body}
\end{subequations}

\noindent The generalized momenta $(\bm{\pi}^1, \bm{\pi}^2, \bm{\pi}^3) = \bm{p}\in T_{\bm{q}}^{*}Q$
are defined as
\begin{equation}
  \mbs{\pi}^i
  =
  \mathscr{J}^{i1}\dot{\bm{d}}_1+\mathscr{J}^{i2}\dot{\bm{d}}_2+\mathscr{J}^{i3}\dot{\bm{d}}_3\, ,
  \label{eq-dir-momenta}
\end{equation}
where Euler's inertia coefficients are
\begin{equation}
  \mathscr{J}^{ij} = \mathscr{J}^{ij}
  =
  \int_{\mathcal{B}_0}\varrho_0\theta^i\theta^j \textrm{d}\mathcal{B}_0 \,,
  \label{eq-euler-inertia}
\end{equation}
for $i$ and $j$ from $1$ to $3$. In addition, the \textit{splitting} of the inertia tensor implies $\mathscr{J}^{13}=\mathscr{J}^{23} = 0$ and $\bm{H}_i \in L(T_{\bm{d}_i}S^2, \reals^n)$
stands for $\frac{\partial \bm{h}}{\partial \bm{d}_i}$.\\

\noindent The pertaining initial conditions are:

\begin{equation}
\bm{d}_1(0)=\bar{\bm{d}}_1\, ,\quad
\bm{d}_2(0)=\bar{\bm{d}}_2\, ,\quad
\bm{d}_3(0)=\bar{\bm{d}}_3\, ,\quad
\bm{w}_\perp(0) = \bar{\bm{d}}_3\times\dot{\bar{\bm{d}}}_3\, ,\quad
w_\parallel(0) = \bar{w}_\parallel\,
\end{equation}
\end{theorem}

\begin{proof}
The theorem follows from the systematic calculation of $\delta S$.	
\end{proof}
The reparametrized equations presented above are totally equivalent to the commonly used equations for the standard rotating rigid body. Consequently, the conservation laws described previously apply directly to this equivalent model. For a in-depth discussion on this subject, the reader may consult \cite{Romero2002}.

\begin{remark}
  As before, to include external moments acting on the standard rotating rigid body,
  the following additional terms need to be added to the variation of the action
	\begin{equation}
	\delta W = \delta W_1+\delta W_2+\delta W_3
	\quad
	\mathrm{with}
	\quad
	\delta W_i = \frac12(\bm{d}_i\times\delta\bm{d}_i)
	\cdot
	(\bm{m}^{\mathrm{ext}}_\perp+ m^{\mathrm{ext}}_\parallel\bm{d}_3)\, .
	\end{equation}
\end{remark}

\section{Rotating rigid body based on the \emph{non-twisting} frame}
\label{sec-nonholonomic-formulation}

In this section, we introduce the nonholonomic rotating rigid body, which incorporates the
non-integrable constraint that is necessary to set the \emph{non-twisting} frame according to
Eq. \eqref{eq-non-twisting}. This is a non-variational model, since it cannot be derived directly
from a variational principle.  For this purpose, we modify Eq. \eqref{eq-standard} to account the
non-integrable condition $W_\parallel=w_\parallel=0$ according to the usual nonholonomic
approach. We also introduce the concomitant conservation laws. Additionally, we present an
alternative formulation that relies on constrained directors, whose particular mathematical
structure enables the application of structure-preserving integration schemes.  

\subsection{Governing equations and invariants}

\begin{theorem}
The nonholonomic equations of motion, \ie Lagrange-D'Alembert equations, for the rotating body model based on the \emph{non-twisting} frame are:	

\begin{subequations}
	\begin{align}
		\covder{\bm{\pi}_{\perp}}{\dot{\bm{d}}_3}+\pi_{\parallel}\dot{\bm{d}}_3 & = \bm{0}\,,\\
		\dot{\pi}_{\parallel}+\bm{\pi}_\perp\cdot\dot{\bm{d}}_3+\mu &= 0\,, \\
		\bm{w}\cdot\bm{d}_3&=0\,.
	\end{align}
	\label{eq-nonholonomic}
\end{subequations}

The pertaining initial conditions are:
\begin{equation}
\bm{\Lambda}(0) = \bar{\bm{\Lambda}}\, ,\quad \bm{w}_\perp(0) = \bar{\bm{w}}_\perp\, ,\quad w_\parallel(0) = 0.    
\end{equation}
Moreover, Eq. \eqref{eq-nonholonomic} can be rewritten as

\begin{subequations}
	\begin{align}
	\covder{\bm{\pi}_{\perp}}{\dot{\bm{d}}_3} & = \bm{0}\, , \\
	\bm{\pi}_\perp\cdot\dot{\bm{d}}_3+\mu &= 0\, , \\
	w_\parallel &=0\, ,
	\end{align}
	\label{eq-nonholonomic-red}
\end{subequations}
where $\bm{\pi}_\perp\in T^{*}_{\bm{d}_3}S^2$ must satisfy the parallel transport along the curve $\bm{d}_3\in S^2$.
\end{theorem}
\begin{proof}
The first part follows from the inclusion of the force associated to the presence of the nonholonomic restriction given by
\begin{equation}
	g = \bm{w}\cdot\bm{d}_3 = 0
\end{equation}
which ensures that the rotating frame renders no twist at all. The virtual work performed by the force associated to the presence of this nonholonomic restriction can be computed as
\begin{equation}
	\delta \mathcal{W}_\mathrm{nh} =
	\mu\frac{\partial\left(\bm{w}\cdot\bm{d}_3\right)}{\partial\bm{w}}\cdot\delta\bm{\theta} =
	\delta\bm{\theta}\cdot\left(\mu\bm{d}_3\right) =
	\delta\bm{\theta}_\parallel\cdot\left(\mu\bm{d}_3\right) \,
\end{equation}
where 
$\mu \in \reals$ denotes the corresponding Lagrange multiplier.
The second part follows from noticing that $w_\parallel=\bm{w}\cdot\bm{d}_3=0$ implies $\pi_\parallel=0$.
\end{proof}

\begin{theorem}
	The conservation laws of the rotating rigid body based on the \emph{non-twisting} frame are:
	\begin{subequations}
		\begin{align}
		K     &= \frac{1}{2}\bm{W}\cdot\bm{J}\bm{W} = \frac{1}{2}\bm{w}\cdot\bm{j}\bm{w} = \text{const.}\, ,\\
		\Pi^1 &= \bm{D}_1\cdot\bm{\Pi} = \bm{d}_1\cdot\bm{\pi} = \text{const.}\, ,\\
		\Pi^2 &= \bm{D}_2\cdot\bm{\Pi} = \bm{d}_2\cdot\bm{\pi} = \text{const.}\, .		
		\end{align}
	\end{subequations}
\end{theorem}

\begin{proof}
 To prove the conservation of kinetic energy, let us consider the following equilibrium statement
\begin{equation}
	\delta \bm{\theta}_\perp\cdot\left(\covder{\bm{\pi}_{\perp}}{\dot{\bm{d}}_3}+\pi_{\parallel}\dot{\bm{d}}_3\right)+
	\delta\bm{\theta}_\parallel\cdot\left( \left(\bm{\pi}_\perp\cdot\dot{\bm{d}}_3+\mu\right)\bm{d}_3\right)+
	\delta \mu \left( \bm{w}\cdot\bm{d}_3 \right)
	= 0\, ,
\end{equation}
where $\delta\bm{\theta}_\perp \in T_{\bm{d}_3}S^2$, $\delta\bm{\theta}_\parallel \in N_{\bm{d}_3}S^2$ and $\delta \mu$ are admissible variations. Now by choosing $\delta \bm{\theta}_\perp = \bm{w}_\perp$, $\delta \bm{\theta}_\parallel = \bm{w}_\parallel$ and $\delta \mu = 0$, we have that 
\begin{equation}
\begin{split}
0 & = 
\bm{w}_\perp\cdot\left(\covder{\bm{\pi}_{\perp}}{\dot{\bm{d}}_3}+\pi_{\parallel}\dot{\bm{d}}_3\right)+
\bm{w}_\parallel\cdot\left( \left(\bm{\pi}_\perp\cdot\dot{\bm{d}}_3+\mu\right)\bm{d}_3\right)\\
& =\bm{w}\cdot\dot{\bm{\pi}} \\
& = \bm{W}\cdot \dot{\bm{\Pi}}\\
& =\pd{}{t}\left(\frac{1}{2}\bm{W}\cdot\bm{J}\bm{W}\right)\\ 
&= \dot{K}\, ,
 \end{split}
 \label{eq-energy_preservation}
\end{equation}
which shows that the kinetic energy is preserved by the motion.

To prove the conservation of the first and second components of the material angular momentum, let us consider the fact that 
\begin{equation}
\begin{split}
\bm{\pi}_\perp &= (\bm{d}_1\cdot\bm{\pi})\bm{d}_1+(\bm{d}_2\cdot\bm{\pi})\bm{d}_2\\
			   &= \Pi^1\bm{d}_1+\Pi^2\bm{d}_2\, .		   
\end{split}
\label{eq-pi-perp}			   
\end{equation}
Now by introducing the former expression into the first statement of Eq. \eqref{eq-nonholonomic-red}, we have that   
\begin{equation}
\begin{split}
	\bm{0} &= \covder{\bm{\pi}_{\perp}}{\dot{\bm{d}}_3}\\
		   &= \covder{\Pi^1\bm{d}_1+\Pi^2\bm{d}_2}{\dot{\bm{d}}_3}\\
   		   &= \covder{\Pi^1\bm{d}_1}{\dot{\bm{d}}_3}+\covder{\Pi^2\bm{d}_2}{\dot{\bm{d}}_3}\\
		   &= \dot{\Pi}^1\bm{d}_1+\dot{\Pi}^2\bm{d}_2\, ,
\end{split}
\label{eq-material_perpendicular_momenta_preservation}
\end{equation}
in which the parallel transport of $\bm{d}_1$ and $\bm{d}_2$, both in $T_{\bm{d}_3}S^2$, has been
accounted for. This shows that the first and second components of the angular momentum are preserved by the motion.
\end{proof}

\begin{remark}
	To include external moments acting on the rotating rigid body based on the \emph{non-twisting} frame, it is necessary to compute the associated virtual work as follows
	\begin{equation}
	\delta W = \delta \bm{\theta}_\perp\cdot\bm{m}^{\mathrm{ext}}_\perp\, .
	\end{equation}
\end{remark}

\subsection{Alternative governing equations}

Here, we present an alternative formulation for the rotating rigid body based on the \emph{non-twisting} frame that relies on constrained directors. The extension of the standard rotating rigid body model to the one relying on the \emph{non-twisting} frame requires the introduction of the constraint given by
\begin{equation}
	g = (1-a)\dot{\bm{d}}_1\cdot\bm{d}_2-a\dot{\bm{d}}_2\cdot\bm{d}_1 = 0\, ,
\end{equation}
in which $a\in [0,1]$ is a parameter that can be freely chosen for convenience. This will be used later on for the proof of the conservation properties of the specialized structure preserving algorithm.
\begin{theorem}
	The alternative nonholonomic equations of motion, \ie the Lagrange-D'Alembert equations,
	for the rotating rigid body model based on the \emph{non-twisting} frame are:
	
	\begin{subequations}
		\begin{align}
		\dot{\bm{\pi}}^1(\bm{d}_1,\bm{d}_2,\bm{d}_3)+\bm{H}_1(\bm{d}_1,\bm{d}_2,\bm{d}_3)^{T}\bm{\lambda}+(1-a)\mu\bm{d}_2&=\bm{0}\,,\\
		\dot{\bm{\pi}}^2(\bm{d}_1,\bm{d}_2,\bm{d}_3)+\bm{H}_2(\bm{d}_1,\bm{d}_2,\bm{d}_3)^{T}\bm{\lambda}-a\mu\bm{d}_1&=\bm{0}\,,\\
		\dot{\bm{\pi}}^3(\bm{d}_1,\bm{d}_2,\bm{d}_3)+\bm{H}_3(\bm{d}_1,\bm{d}_2,\bm{d}_3)^{T}\bm{\lambda}&=\bm{0}\,,\\
		\mbs{h}(\mbs{d}_1,\mbs{d}_2,\mbs{d}_3) &= \bm{0}\,,\\
				(1-a)\dot{\bm{d}}_1\cdot\bm{d}_2-a\dot{\bm{d}}_2\cdot\bm{d}_1 &= 0\, .
		\end{align}
		\label{eq_weak_form_rigid_body_alternative}
	\end{subequations}
	
	\noindent The pertaining initial conditions are:
	
	\begin{equation}
	\bm{d}_1(0)=\bar{\bm{d}}_1\, ,\quad
	\bm{d}_2(0)=\bar{\bm{d}}_2\, ,\quad
	\bm{d}_3(0)=\bar{\bm{d}}_3\, ,\quad
	\bm{w}_\perp(0) = \bar{\bm{d}}_3\times\dot{\bar{\bm{d}}}_3\, ,\quad
	w_\parallel(0) = 0\,
	\end{equation}
\end{theorem}

\begin{proof}
The theorem follows from the computation of the virtual work associated to the presence of the nonholonomic restriction, namely
\begin{equation}
\delta \mathcal{W}_\mathrm{nh} = \mu \frac{\partial g}{\partial \dot{\bm{d}}_1}\cdot\delta \bm{d}_1 +
\mu \frac{\partial g}{\partial \dot{\bm{d}}_2}\cdot\delta \bm{d}_2 = \delta\bm{d}_1\cdot((1-a)\mu\bm{d}_2)+\delta\bm{d}_2\cdot(-a\mu\bm{d}_1)\, ,
\end{equation}
where $\mu \in \reals$ denotes the corresponding Lagrange multiplier.
\end{proof}

\begin{remark}
	To include external moments acting on the rotating rigid body based on the \emph{non-twisting} frame, it is necessary to compute the associated virtual work as follows
	\begin{equation}
	\delta W = \delta W_1+\delta W_2+\delta W_3
	\quad
	\mathrm{with}
	\quad
	\delta W_i = \frac12(\bm{d}_i\times\delta\bm{d}_i)
	\cdot\bm{m}^{\mathrm{ext}}_\perp\, .
	\end{equation}
\end{remark}

\section{Structure-preserving time integration}
\label{sec-integration}

A fundamental aspect to produce acceptable numerical
results in the context of nonlinear systems is the preservation
of mechanical invariants whenever possible, \eg first integrals
of motion. These conservation
properties ensure that beyond the approximation errors,
the computed solution remains consistent with respect to the
underlying physical essence. Here then, we chose the family of integration methods
that is derived by direct discretization of the equations of
motion.

\subsection{Basic energy-momentum algorithm}

Next, we describe the application of the energy-momentum integration algorithm \cite{Simo1991,Simo1992} to the ``standard rotating rigid body'' case. For this purpose, the following nomenclature is necessary:

\begin{equation}
\bm{q}\!=\!
\begin{pmatrix}
\bm{d}_1\\
\bm{d}_2\\
\bm{d}_3\\
\end{pmatrix},\:\:
\bm{p}\!=\!
\begin{pmatrix}
\bm{\pi}^1\\
\bm{\pi}^2\\
\bm{\pi}^3\\
\end{pmatrix}\:\:\:\textrm{and}\:\:\:
\bm{Q}^\textrm{ext}\!=\!
\begin{pmatrix}
\bm{f}^{1, \textrm{ext}}\\
\bm{f}^{2, \textrm{ext}}\\
\bm{f}^{3, \textrm{ext}}\\
\end{pmatrix}.
\end{equation}
While $\bm{q}$ is the vector of generalized coordinates, $\bm{p}$ collects the generalized momenta and $\bm{Q}^\textrm{ext}$ contains the generalized external loads, if present. The discrete version of Eq. (\ref{eq_weak_form_rigid_body}) can be expressed at time $n+\frac{1}{2}$ as
\begin{equation}
\bigl\langle \delta \bm{q}_{n+\frac{1}{2}},
\dot{\bm{p}}_d(\bm{q}_n,\dot{\bm{q}}_n,\bm{q}_{n+1})-\bm{Q}^\textrm{ext}_{n+\frac{1}{2}}+\bm{H}_d^T(\bm{q}_n,\bm{q}_{n+1}) \bm{\lambda}_{n+\frac{1}{2}}
\bigr\rangle+
\bigl\langle \delta \bm{\lambda}_{n+1}, \bm{h}(\bm{q}_{n+1})
\bigr\rangle
= 0\,,
\end{equation}
where $\left\langle\cdot,\cdot\right\rangle$ stands for the dual pairing.

A key point to achieve the desired preservation properties, is to define the momentum terms by using the midpoint rule, \ie
\begin{subequations}
	\begin{align}
		\bm{p}_d(\bm{q}_n,\bm{q}_{n+1})&=\frac{1}{h}\bm{M}(\bm{q}_{n+1}-\bm{q}_n),\\
		\dot{\bm{p}}_d(\bm{q}_n,\dot{\bm{q}}_n,\bm{q}_{n+1})&=\frac{2}{h^2}\bm{M}(\bm{q}_{n+1}-\bm{q}_{n}-h\dot{\bm{q}}_{n})\, ,
	\end{align}
\end{subequations}
where $\bm{q}_n$ and $\dot{\bm{q}}_n$ are known from the previous step, $\bm{q}_{n+1}$ are $\dot{\bm{q}}_{n+1}$ are unknown, and $\dot{\bm{q}}_{n+1}$ is computed as $\frac{2}{h}(\bm{q}_{n+1}-\bm{q}_n)-\dot{\bm{q}}_n$ once $\bm{q}_{n+1}$ has been determined by means of an iterative procedure, typically the Newtown-Raphson method.\\

The mass matrix takes the form
\begin{equation}
\bm{M}=
\begin{bmatrix}
\mathscr{J}^{11}\bm{I}_{3\times 3} & \mathscr{J}^{12}\bm{I}_{3\times 3} & \bm{0}_{3\times 3} \\
\mathscr{J}^{12}\bm{I}_{3\times 3} & \mathscr{J}^{22}\bm{I}_{3\times 3} & \bm{0}_{3\times 3} \\
\bm{0}_{3\times 3} & \bm{0}_{3\times 3} & \mathscr{J}^{33}\bm{I}_{3\times 3} \\
\end{bmatrix}
\end{equation}
and $\mathscr{J}^{ij}$ for $i$ and $j$ running from $1$ to $3$ being defined above. This very simple construction satisfies, only for the standard rigid body case, the preservation of linear and angular momenta in combination with the kinetic energy in absence of external loads.\\

The discrete version of the Jacobian matrix of the constraints can be computed with the \textit{average vector field} \cite{Gebhardt2019a,Gebhardt2019b} as
\begin{equation}
\bm{H}_d(\bm{q}_n,\bm{q}_{n+1})=\frac{1}{2}\int_{-1}^{+1} \left.\frac{\partial \bm{h}}{\partial \bm{q}} \right|_{\bm{q}(\xi)} \mathrm{d}\xi,
\end{equation}
where $\bm{q}(\xi)$ is defined as $\frac{1}{2}(1-\xi)\bm{q}_{n}+\frac{1}{2}(1+\xi)\bm{q}_{n+1}$ for $\xi\in[-1,+1]$. The algorithmic Jacobian matrix defined in this way satifies for any admissible solution the discrete version of the hidden constraints, \ie 
\begin{equation}
\bm{H}_d(\bm{q}_{n}, \bm{q}_{n+1})(\bm{q}_{n+1}-\bm{q}_{n}) =\bm{0}.
\end{equation}

\begin{theorem} The discrete conservation laws of the energy-momentum integration algorithm specialized to the standard rotating rigid body are:
	\begin{subequations}
	\begin{align}
	K_{n+1} - K_n &= 0 \, ,\\
	\bm{\pi}_{n+1} - \bm{\pi}_n &= 0\, .         
	\end{align}
	\end{subequations}
\end{theorem}
\begin{proof}
	This is an standard result and thus, we omit further details.
\end{proof}

\subsection{Specialized energy-momentum algorithm}
\noindent The energy-momentum integration algorithm can be further specialized to the nonholonomic case, where the discrete governing equations are:

\begin{equation}
\begin{gathered}
\bigl\langle \delta \bm{q}_{n+\frac{1}{2}},
\dot{\bm{p}}_d(\bm{q}_n,\dot{\bm{q}}_n,\bm{q}_{n+1})-\bm{Q}^\textrm{ext}_{n+\frac{1}{2}}+\\
\bm{H}_d^T(\bm{q}_n,\bm{q}_{n+1}) \bm{\lambda}_{n+\frac{1}{2}}+\bm{G}_d^T(\bm{q}_n,\bm{q}_{n+1})\bm{\mu}_{n+\frac{1}{2}} 
\bigr\rangle+
\bigl\langle \delta \bm{\lambda}_{n+1}, \bm{h}(\bm{q}_{n+1})
\bigr\rangle
= 0\, ,\\
\bm{G}_d(\bm{q}_n,\bm{q}_{n+1})(\bm{q}_{n+1}-\bm{q}_n) = \bm{0}\, .
\end{gathered}
\label{eq-discrete-nonholonomic}
\end{equation}

Once again, we can use the \textit{average vector field} to compute
\begin{equation}
\bm{G}_d(\bm{q}_n,\bm{q}_{n+1})=\frac{1}{2}\int_{-1}^{+1} \bm{G}(\bm{q}(\xi)) \mathrm{d}\xi
\end{equation}
that arises from the nonholonomic constraint, where $\bm{q}(\xi)$ is defined as before. In this way, the nonholonomic constraint is identically satisfied at the midpoint, \ie 
\begin{equation}
\bm{G}_d(\bm{q}_{n}, \bm{q}_{n+1})(\bm{q}_{n+1}-\bm{q}_{n}) =\bm{0}\, .
\end{equation}
\begin{theorem} The discrete conservation laws of the energy-momentum integration algorithm specialized to the rotating rigid body based on the \emph{non-twisting} frame are:
	\begin{subequations}
		\begin{align}
		K_{n+1} - K_n &= 0 \, ,\\
		\Pi^1_{n+1} - \Pi^1_n &= 0\, ,\\
		\Pi^2_{n+1} - \Pi^2_n &= 0\, .
		\end{align}
	\end{subequations}
\end{theorem}
\begin{proof}
To prove the conservation of kinetic energy, let us consider the following discrete variation
\begin{equation}
(\delta \bm{q}_{n+\frac{1}{2}}, \delta \bm{\lambda}_{n+1}) = \frac{h}{2}(\dot{\bm{q}}_{n+1}+\dot{\bm{q}}_n, \bm{0})\, .
\end{equation}

 By inserting the previous discrete variation in Eq. \eqref{eq-discrete-nonholonomic}, we get  
\begin{equation}
\begin{split}
0 &= \frac{1}{2}(\dot{\bm{q}}_{n+1}+\dot{\bm{q}}_n)\cdot \bm{M}(\dot{\bm{q}}_{n+1}-\dot{\bm{q}}_n) \\
  & = \frac{1}{2}\dot{\bm{q}}_{n+1}\cdot \bm{M}\dot{\bm{q}}_{n+1}-\frac{1}{2}\dot{\bm{q}}_n\cdot \bm{M}\dot{\bm{q}}_n \\
  &= K_{n+1}-K_n\, .\\
\end{split}
\end{equation}

 For the first component of the angular momentum, \ie $\Pi^1$, we need to consider the following discrete variation

\begin{equation}
\begin{split}
(\delta \bm{q}_{n+\frac{1}{2}}, \delta \bm{\lambda}_{n+1}) & = (\delta \bm{d}_{1, n+\frac{1}{2}}, \delta \bm{d}_{2, n+\frac{1}{2}}, \delta \bm{d}_{3, n+\frac{1}{2}}, \delta \bm{\lambda}_{n+1}) \\
& =  \frac{h}{2}( \bm{0}, \bm{d}_{3, n+1}+\bm{d}_{3, n}, -\bm{d}_{2, n+1}-\bm{d}_{2, n}, \bm{0}, \bm{0})\, , 
\end{split}
\end{equation}
and let $a$ be equal to $1$. By inserting the previous discrete variation in Eq. \eqref{eq-discrete-nonholonomic}, we get  
\begin{equation}
\frac{1}{2}(\bm{d}_{3, n+1}+\bm{d}_{3, n}) \cdot (\bm{\pi}^2_{n+1}-\bm{\pi}^2_n)-\frac{1}{2}(\bm{d}_{2, n+1}+\bm{d}_{2, n}) \cdot (\bm{\pi}^3_{n+1}-\bm{\pi}^3_n) = 0\, ,
\label{eq-discrete-nonholonomic-1}
\end{equation}
where
\begin{equation}
(\bm{d}_{3, n+1}+\bm{d}_{3, n}) \cdot (\bm{\pi}^2_{n+1}-\bm{\pi}^2_n) = 
\bm{d}_{3, n+1}\cdot\bm{\pi}^2_{n+1}-\bm{d}_{3, n}\cdot\bm{\pi}^2_{n}+\bm{d}_{3, n}\cdot\bm{\pi}^2_{n+1}-\bm{d}_{3, n+1}\cdot\bm{\pi}^2_{n}\, .
\end{equation}
By using Taylor's approximations, we have that 
\begin{equation}
\bm{d}_{3, n+1}\cdot\bm{\pi}^2_{n} = \bm{d}_{3, n}\cdot\bm{\pi}^2_{n}+h(\dot{\bm{d}}_{3, n}\cdot\bm{\pi}^2_{n})+\Order{h}{2}\,
\end{equation}
and
\begin{equation}
\bm{d}_{3, n}\cdot\bm{\pi}^2_{n+1} = \bm{d}_{3, n+1}\cdot\bm{\pi}^2_{n+1}-h(\dot{\bm{d}}_{3, n+1}\cdot\bm{\pi}^2_{n+1})+\Order{h}{2}\, .
\end{equation}
Then
\begin{equation}
\begin{split}
	\bm{d}_{3, n}\cdot\bm{\pi}^2_{n+1}-\bm{d}_{3, n+1}\cdot\bm{\pi}^2_{n} 
	 & = \bm{d}_{3, n+1}\cdot\bm{\pi}^2_{n+1}-\bm{d}_{3, n}\cdot\bm{\pi}^2_{n}+h(\dot{\bm{d}}_{3,n}\cdot\bm{\pi}^2_{n+1}+\dot{\bm{d}}_{3,n+1}\cdot\bm{\pi}^2_{n})+\Order{h}{2}\\
	& = \bm{d}_{3, n+1}\cdot\bm{\pi}^2_{n+1}-\bm{d}_{3, n}\cdot\bm{\pi}^2_{n}+\Order{h}{2}\, ,
\end{split}	
\end{equation}
insomuch as
\begin{equation}
	\dot{\bm{d}}_{3,n+1}\cdot\bm{\pi}^2_{n}+\dot{\bm{d}}_{3,n}\cdot\bm{\pi}^2_{n+1} 
	= 2\dot{\bm{d}}_{3,n+\frac{1}{2}}\cdot\bm{\pi}^2_{n+\frac{1}{2}}+\Order{h}{}
	= \Order{h}{}\, ,
\end{equation}
which can be easily shown by considering that
\begin{equation}
\begin{split}
\dot{\bm{d}}_{3}\cdot\bm{\pi}^2 
&=\dot{\bm{d}}_{3}\cdot(\mathscr{J}^{12}\dot{\bm{d}}_{1}+\mathscr{J}^{22}\dot{\bm{d}}_{2})\\
&=\mathscr{J}^{12}\dot{\bm{d}}_{3}\cdot\dot{\bm{d}}_{1}+\mathscr{J}^{22}\dot{\bm{d}}_{3}\cdot\dot{\bm{d}}_{2}\\
& = 0\, ,
\end{split}
\end{equation}
since the angular velocity $\bm{w}_{n+\frac{1}{2}}$ has the form $\bm{d}_{3,n+\frac{1}{2}}\times\dot{\bm{d}}_{3,n+\frac{1}{2}}$ due to the satisfaction of the \emph{non-twisting} condition, see Subsection \ref{subs-non-twisting-frame}.\\

 In this way, the first term of Eq. \eqref{eq-discrete-nonholonomic-1} becomes	
\begin{equation}
\frac{1}{2}(\bm{d}_{3, n+1}+\bm{d}_{3, n}) \cdot (\bm{\pi}^2_{n+1}-\bm{\pi}^2_n) 
= \bm{d}_{3, n+1}\cdot\bm{\pi}^2_{n+1}-\bm{d}_{3, n}\cdot\bm{\pi}^2_{n}+\Order{h}{2}\, ,
\end{equation}
and with same reasoning, the second term of Eq. \eqref{eq-discrete-nonholonomic-1} turns to be
\begin{equation}
\frac{1}{2}(\bm{d}_{2, n+1}+\bm{d}_{2, n}) \cdot (\bm{\pi}^3_{n+1}-\bm{\pi}^3_n) 
= \bm{d}_{2, n+1}\cdot\bm{\pi}^3_{n+1}-\bm{d}_{2, n}\cdot\bm{\pi}^3_{n}+\Order{h}{2}\,.
\end{equation}
By replacing the two previous expressions in \eqref{eq-discrete-nonholonomic-1}, we have that
\begin{equation}
\begin{split}
	0 &= (\bm{d}_{3, n+1}\cdot\bm{\pi}^2_{n+1}-\bm{d}_{2, n+1}\cdot\bm{\pi}^3_{n+1})-(\bm{d}_{3, n}\cdot\bm{\pi}^2_{n}-\bm{d}_{2, n}\cdot\bm{\pi}^3_{n})+\Order{h}{2}\\
	  &= \Pi^1_{n+1}-\Pi^1_{n}+\Order{h}{2}\\
	  &\approx \Pi^1_{n+1}-\Pi^1_{n}\, ,
\end{split}
\end{equation}
which is true since
\begin{equation}
\begin{split}
\Pi^1 &= \bm{d}_1\cdot\bm{\pi}\\
&= \bm{d}_1\cdot\left(\bm{d}_1\times\bm{\pi}^1+\bm{d}_2\times\bm{\pi}^2+\bm{d}_3\times\bm{\pi}^3\right) \\
&= \bm{d}_3\cdot \bm{\pi}^2-\bm{d}_2\cdot \bm{\pi}^3\, .
\end{split}
\end{equation}
Finally, for the second component of the angular momentum, \ie $\Pi^2$, we need to consider the following discrete variation
\begin{equation}
\begin{split}
(\delta \bm{q}_{n+\frac{1}{2}}, \delta \bm{\lambda}_{n+1}) & = (\delta \bm{d}_{1, n+\frac{1}{2}}, \delta \bm{d}_{2, n+\frac{1}{2}}, \delta \bm{d}_{3, n+\frac{1}{2}}, \delta \bm{\lambda}_{n+1}) \\
& =  \frac{h}{2}(-\bm{d}_{3, n+1}-\bm{d}_{3, n}, \bm{0}, \bm{d}_{2, n+1}+\bm{d}_{2, n}, \bm{0}, \bm{0})\,
\end{split}
\end{equation}
and let $a$ be equal to zero. Then, the rest of the proof follows as before.
\end{proof}


\section{Numerical results}
\label{sec-results}

In this section, we present numerical results of the motion of a rotating rigid body based on the \emph{non-twisting} frame with (non-physical) inertia

\begin{equation}
[J_{ij}] =
\begin{bmatrix}
3 & -1/7 & 0 \\
-1/7 &    4 & 0 \\
0 &    0 & 5
\end{bmatrix}
\,\, \mathrm{~Kg\,m^2}
\,\,
\textrm{or equivalently}
\,\,
[\mathscr{J}_{ij}] =
\begin{bmatrix}
3 &  1/7 & 0 \\
1/7 &    2 & 0 \\
0 &    0 & 1
\end{bmatrix}
\,\, \mathrm{~Kg\,m^2}\, .
\end{equation}
Next, we evaluate the qualitative properties of the proposed numerical setting in a reduced picture. For the first case, we consider the dynamic response to an initial condition different from the trivial one. For the second case, we consider the dynamic response to a vanishing load. The third and last case is a combination of both, \ie initial condition different from the trivial one and a vanishing  load. All the three cases were numerically solved in the time interval $[0, 5]\,\,\textrm{s}$ with a time step size of $h = 0.005\,\,\textrm{s}$ and relative tolerance $10^{-10}$.

\subsection{Case 1 - response to nonzero initial conditions}

For this first case we consider
\begin{equation}
\bm{\Lambda}(0) = \bm{I}\, , \quad \bm{\omega}_\perp(0) = 6\bm{d}_1(0)-18\bm{d}_2(0)
\,\,\mathrm{~rad/s}\,
\end{equation}
and
\begin{equation}
\bm{m}^{\mathrm{ext}}_\perp(t) = \bm{0}
\,\,\mathrm{~Kg\,m^2/s^2}.
\end{equation}
Figure \ref{fig-angular-momentum-ic} presents the time history for the spatial and material
components of the angular momentum. On the left, we can observe that the components of the spatial
angular momentum (SAM) oscillate with constant amplitude and frequency and therefore, they are not
constant as in the case of the standard rotating rigid body. On the right we can observe that the
components of the material angular momentum (MAM) are identically preserved. While the first and second components are constant and different from zero, the third one is zero as expected from the imposition of the nonholonomic restriction $\omega_\parallel = 0\,\,\mathrm{rad/s}$. As shown before for the analytical setting as well as for the numerical setting, this non-intuitive behavior results from the fact that the dynamics of the system is not taking place in the environment space, but on the 2-sphere. Therefore, this behavior is truly native on the 2-sphere since the directors $\bm{d}_1$ and $\bm{d}_2$ in $T_{\bm{d}_3}S^2$ are being parallel transported along the time-parameterized solution curve $\bm{d}_3$.\\

Figure \ref{fig-energy-precision-ic} presents the time history for the kinetic energy and
the second quotient of precision as defined in the Appendix. On the left we can observe that the kinetic energy is identically preserved as expected. On the right we see that the second quotient of precision is approximately $4$, see also Table \ref{tab-sqp_ic}, which means that the integrator is really achieving second-order accuracy.\\

Figure \ref{fig-trajectory-ic} shows the trajectory followed by $\bm{d}_3$, which as expected takes
place on a plane that separates the sphere into two half spheres. Such trajectory minimizes locally
the distance on $S^2$ and thus, this is geodesic. Finally, and to summarize the excellent performance
of the numerical setting, Table \ref{tab-invariants_ic} presents the stationary values for the
motion invariants, \ie the first and second components of the material angular momentum and kinetic energy.

\begin{figure}[h!]
	\centering{}
	\begin{tabular}{cc}
		\includegraphics[width=0.45\textwidth]{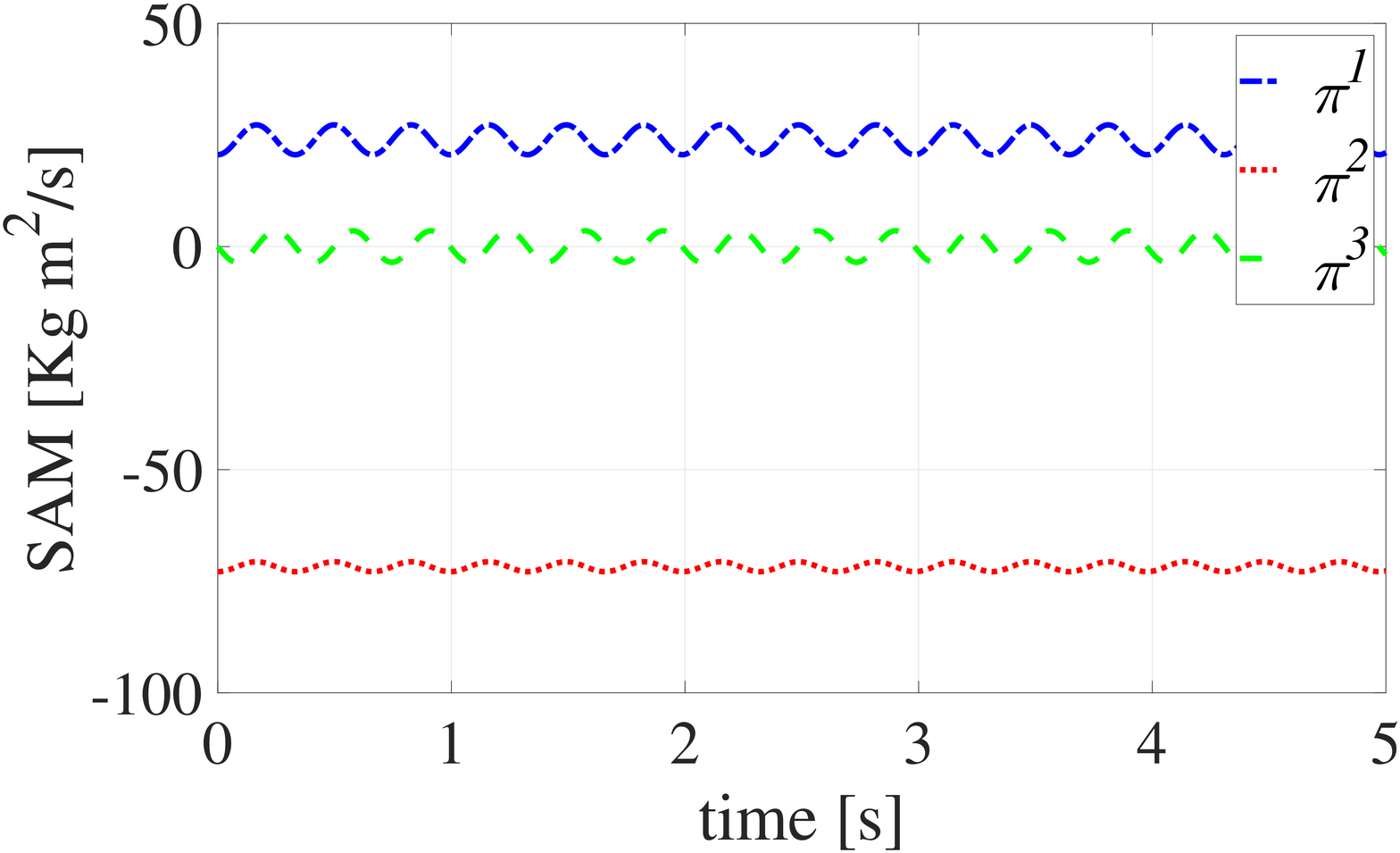} &	\includegraphics[width=0.45\textwidth]{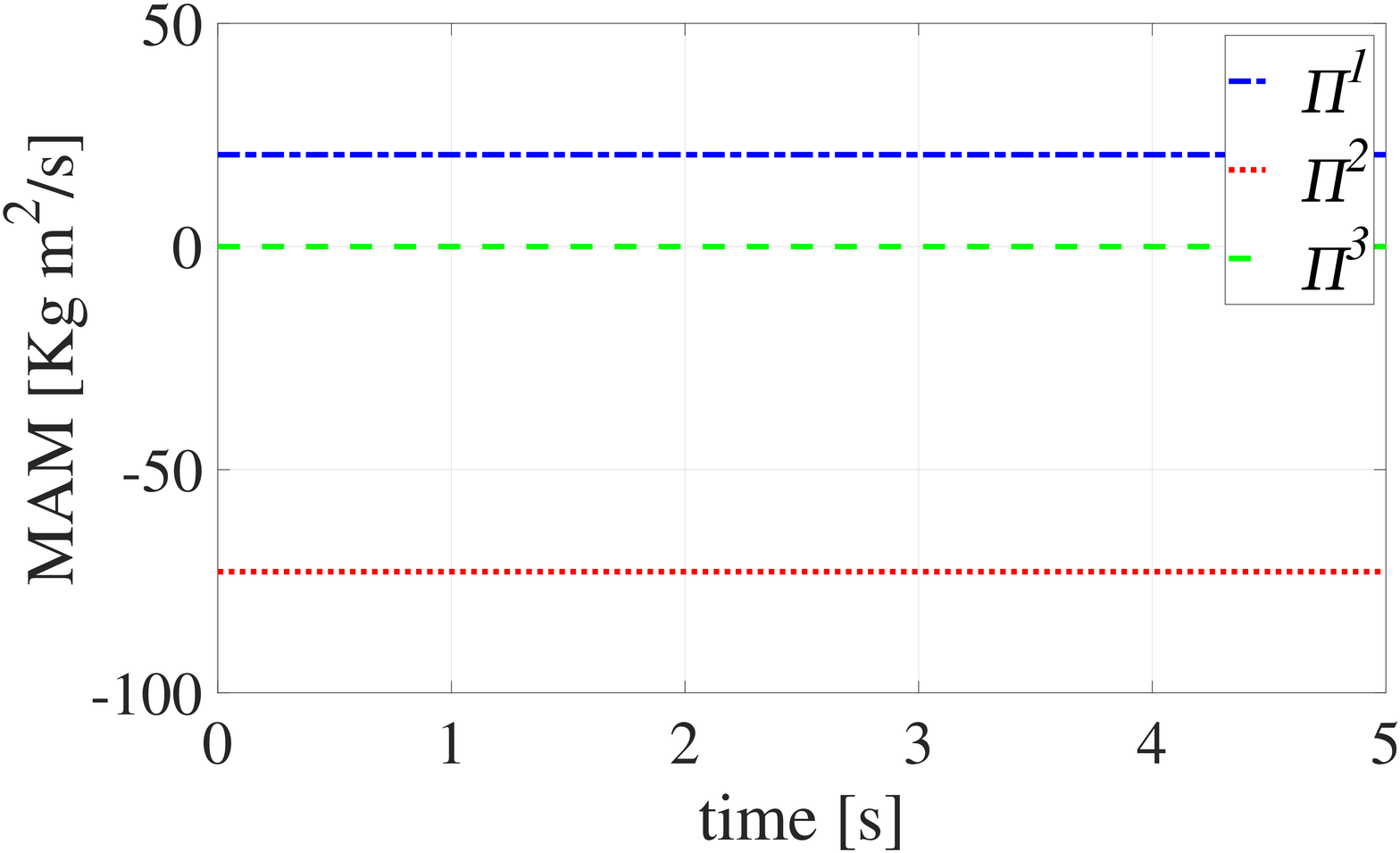}
	\end{tabular}
	\caption{Case 1 - SAM components (left) and MAM components (right).}
	\label{fig-angular-momentum-ic}
\end{figure} 

\begin{figure}[h!]
	\centering{}
	\begin{tabular}{cc}
		\includegraphics[width=0.45\textwidth]{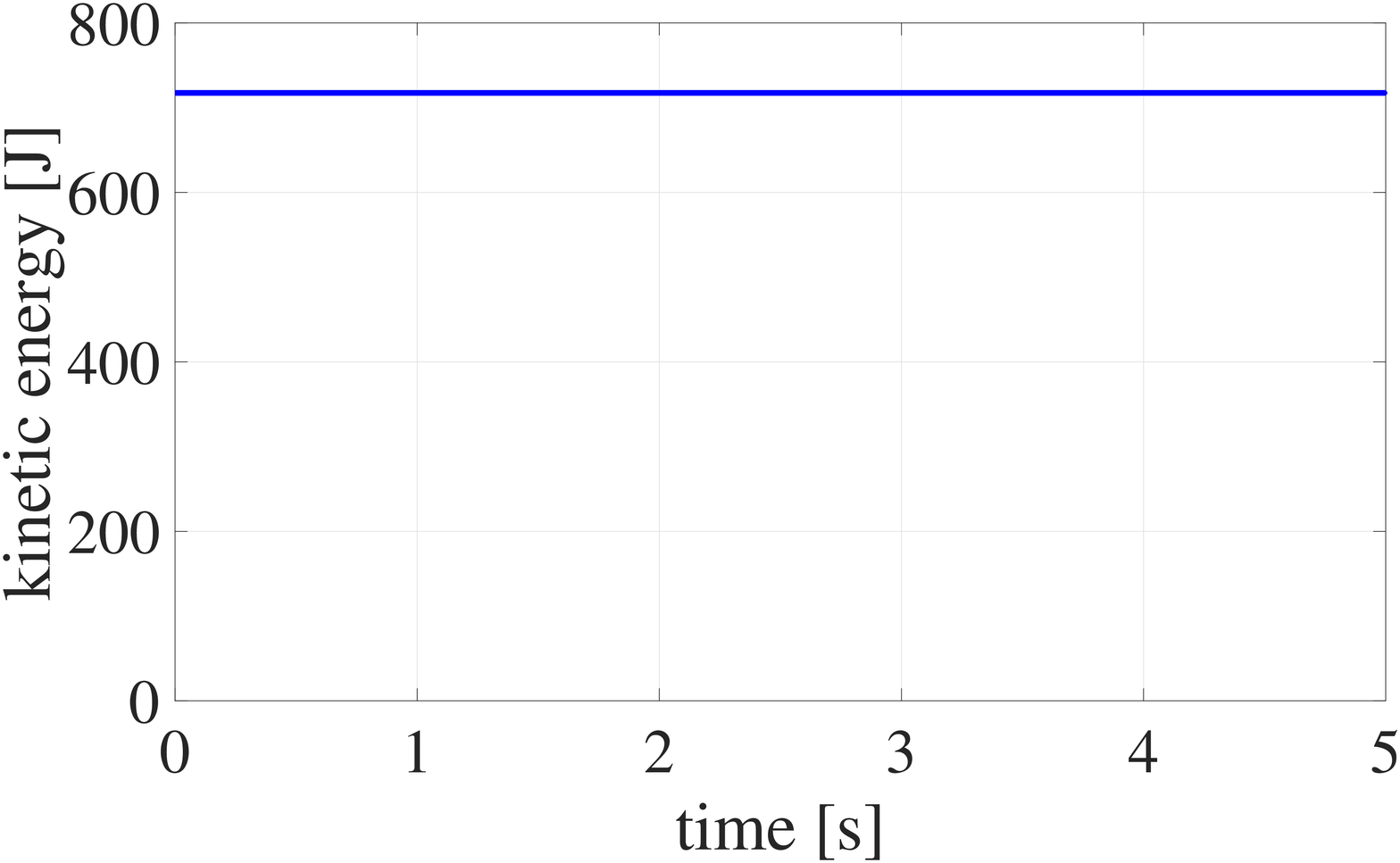} &	\includegraphics[width=0.45\textwidth]{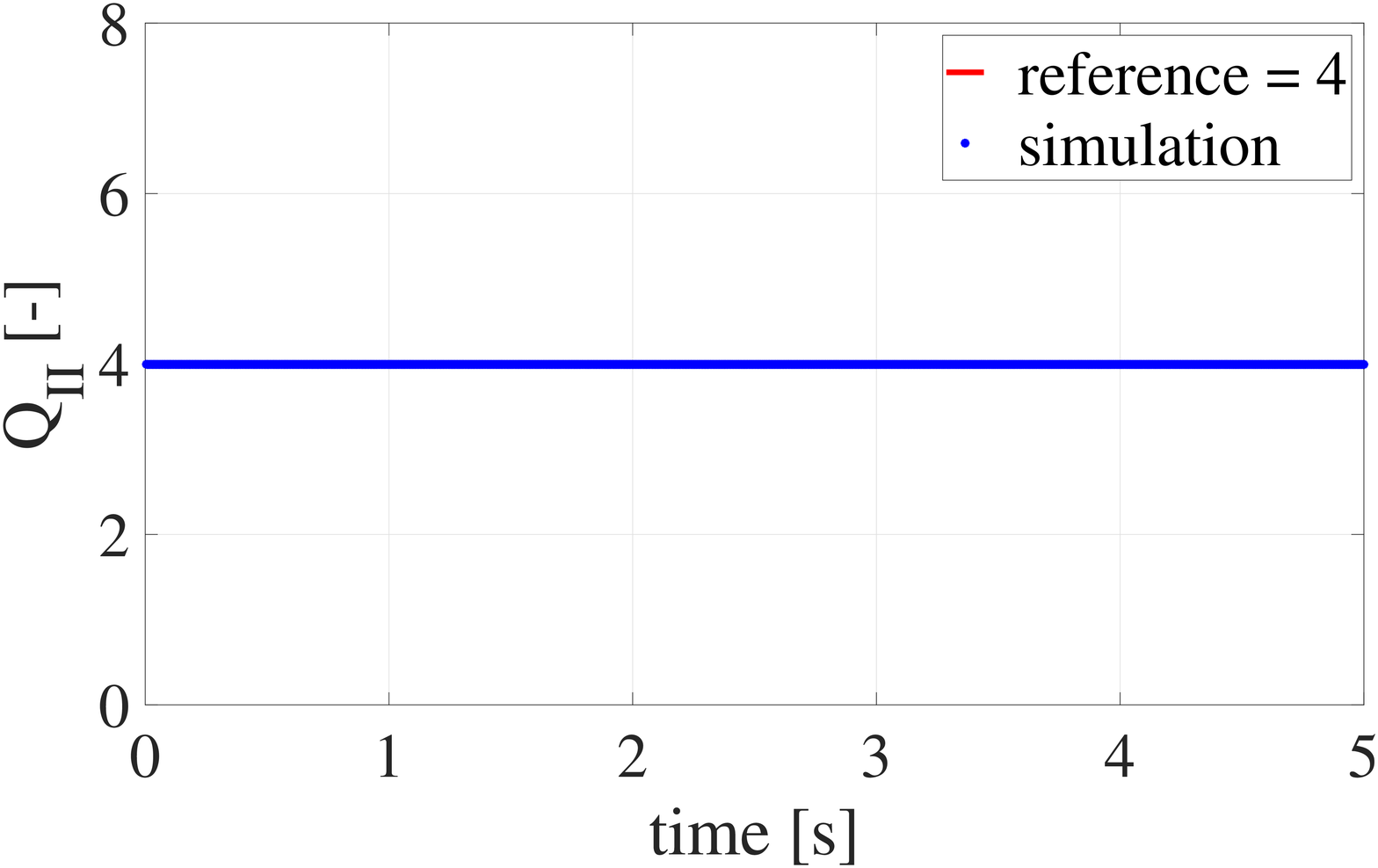}
	\end{tabular}
	\caption{Case 1 - kinetic energy (left) and second quotient of precision (right).}
	\label{fig-energy-precision-ic}
\end{figure} 

\begin{figure}[h!]
	\centering{}
	\includegraphics[width=0.5\textwidth]{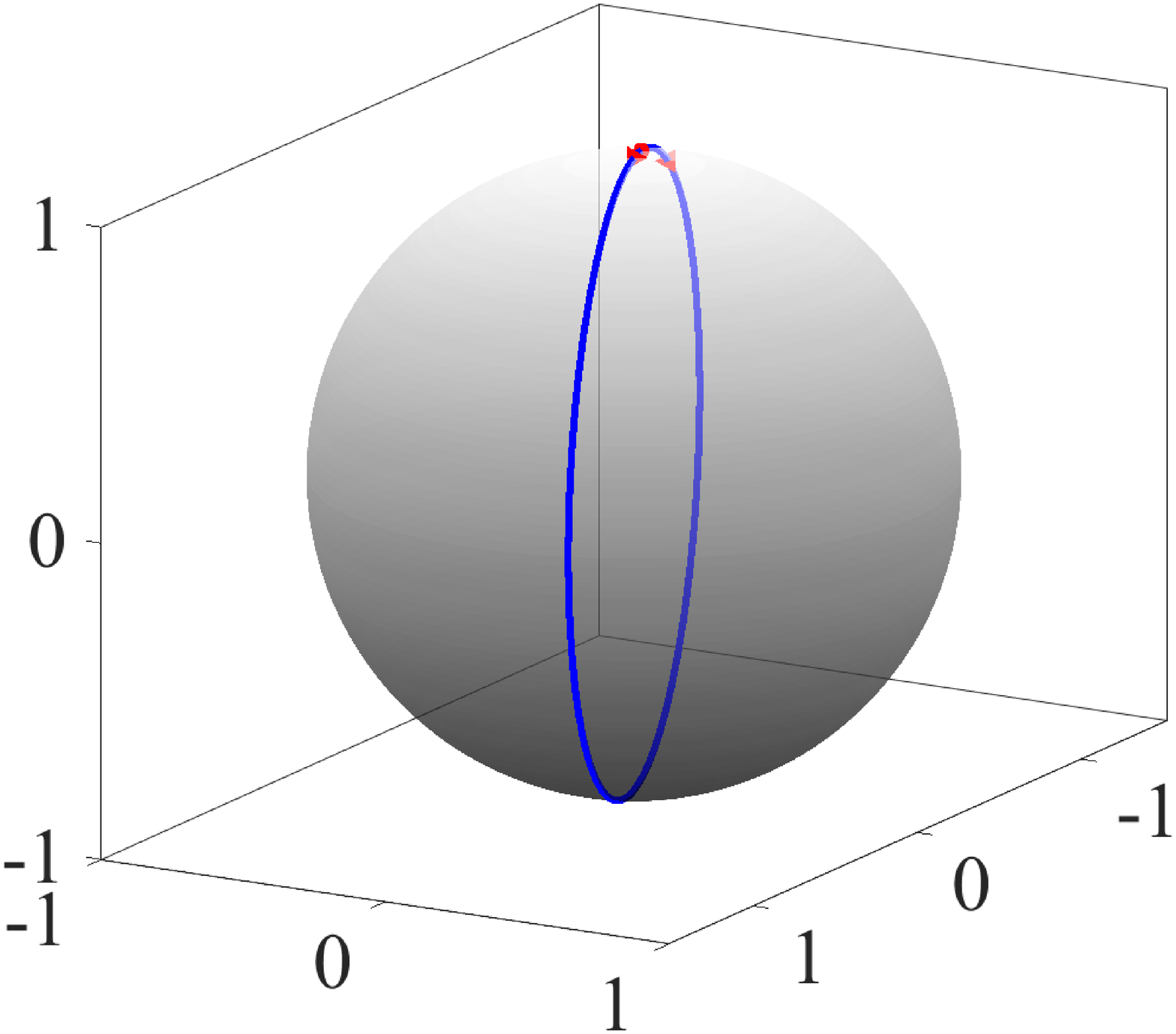}
	\caption{Case 1 - trajectory of $\bm{d}_3$ on $S^2$. }
	\label{fig-trajectory-ic}
\end{figure}

\begin{table}[!ht]
	\centering{}
	\begin{tabular}{|c|c|}
		\hline
		$t$        & $Q_{\mathrm{II}}$ \tabularnewline
		$[$s$]$    &             $[-]$ \tabularnewline
		\hline
		$1.000000$ &        $3.994933$ \tabularnewline
		\hline
		$2.000000$ &        $3.994869$ \tabularnewline
		\hline
		$3.000000$ &        $3.994787$ \tabularnewline
		\hline
		$4.000000$ &        $3.994663$ \tabularnewline
		\hline
		$5.000000$ &        $3.994504$ \tabularnewline		
		\hline		
	\end{tabular}\caption{Case 1 - second quotient of precision.}
	\label{tab-sqp_ic}
\end{table}

\begin{table}[!ht]
	\centering{}
	\begin{tabular}{|c|c|c|c|}
		\hline
		$t$        & $\Pi^1$             & $\Pi^2$             & $K$ \tabularnewline
		$[$s$]$    & $[$Kg$\,$m$^2$/s$]$ & $[$Kg$\,$m$^2$/s$]$ & $[$J$]$ \tabularnewline
		\hline
		$1.000000$ & $20.571428$        & $-72.857142$        & $717.428571$ \tabularnewline
		\hline
		$2.000000$ & $20.571428$        & $-72.857142$        & $717.428571$ \tabularnewline
		\hline
		$3.000000$ & $20.571428$        & $-72.857142$        & $717.428571$ \tabularnewline
		\hline
		$4.000000$ & $20.571428$        & $-72.857142$        & $717.428571$ \tabularnewline
		\hline
		$5.000000$ & $20.571428$        & $-72.857142$        & $717.428571$ \tabularnewline		
		\hline		
	\end{tabular}\caption{Case 1 - motion invariants - stationary values.}
	\label{tab-invariants_ic}
\end{table}

\subsection{Case 2 - response to a vanishing load}

For this second case we consider
\begin{equation}
\bm{\Lambda}(0) = \bm{I}\, , \quad \bm{\omega}_\perp(0) = \bm{0}\,\,\mathrm{~rad/s}
\end{equation}
and
\begin{equation}
\bm{m}^{\mathrm{ext}}_\perp(t) =-f(t)(2472.5\bm{d}_1(t)+1075\bm{d}_2(t))\,\,\mathrm{~ Kgm^2/s^2} ,
\end{equation}
where
\begin{equation}
f(t)=\left\{ \begin{array}{ccc}
2t & \mathrm{for} & 0\leq t<0.5\\
2-2t & \mathrm{for} & 0.5\leq t<1\\
0 & \mathrm{for} & t\geq1
\end{array}\right.
\label{eq-loads_time_variation}
\end{equation}

Figure \ref{fig-angular-momentum-ml} presents the time history for the spatial and material
components of the angular momentum, where the applied material load $\bm{m}^{\mathrm{ext}}_\perp$ is
active only during the first second of simulation. On the left figure, we observe that the
components of the spatial angular momentum vary starting from zero since the rotating rigid body is
initially at rest. After the load vanishes, the components of the spatial angular momentum
oscillate with constant amplitude and frequency and therefore, they are not constant, but indicate a
steady state. On the right figure we can observe that the components of the material angular momentum also
vary from zero, except the third one that remains always equal to zero. After the material load
vanishes, the components of the material angular momenta  are identically preserved. Once again, the
first and second components are constant and different from zero.

Figure \ref{fig-energy-precision-ml} presents the time history for the kinetic energy and the second
quotient of precision. On the left, we can observe that the kinetic energy vary during the first
second, where the applied material load is active. After this vanishes, the kinetic energy is
identically preserved. On the right figure we confirm again that the second quotient of precision is
approximately $4$, see also Table \ref{tab-sqp_ml}, which means that the integrator is second-order
accurate even during the time in which the applied material load is active.

Figure \ref{fig-trajectory-ml} shows the trajectory followed by $\bm{d}_3$, which due to the fixed relation among components of the applied material load  takes place on a plane that separates the sphere in two half spheres. Such trajectory minimizes locally the distance on $S^2$ and thus, this is geodesic as well. Table \ref{tab-invariants_ml} presents the stationary values for the motion invariants.

\begin{figure}[h!]
	\centering{}
	\begin{tabular}{cc}
		\includegraphics[width=0.45\textwidth]{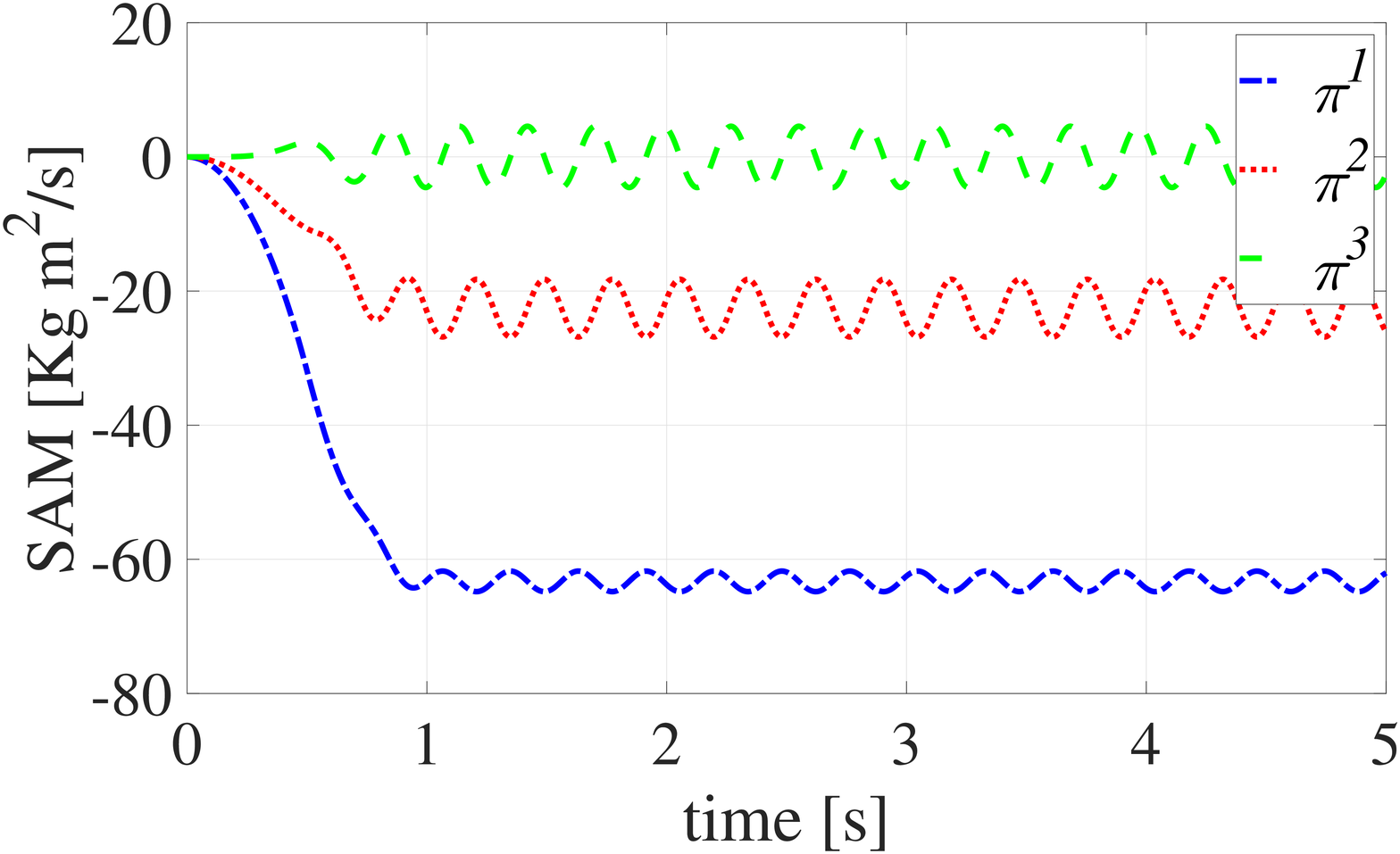} &	\includegraphics[width=0.45\textwidth]{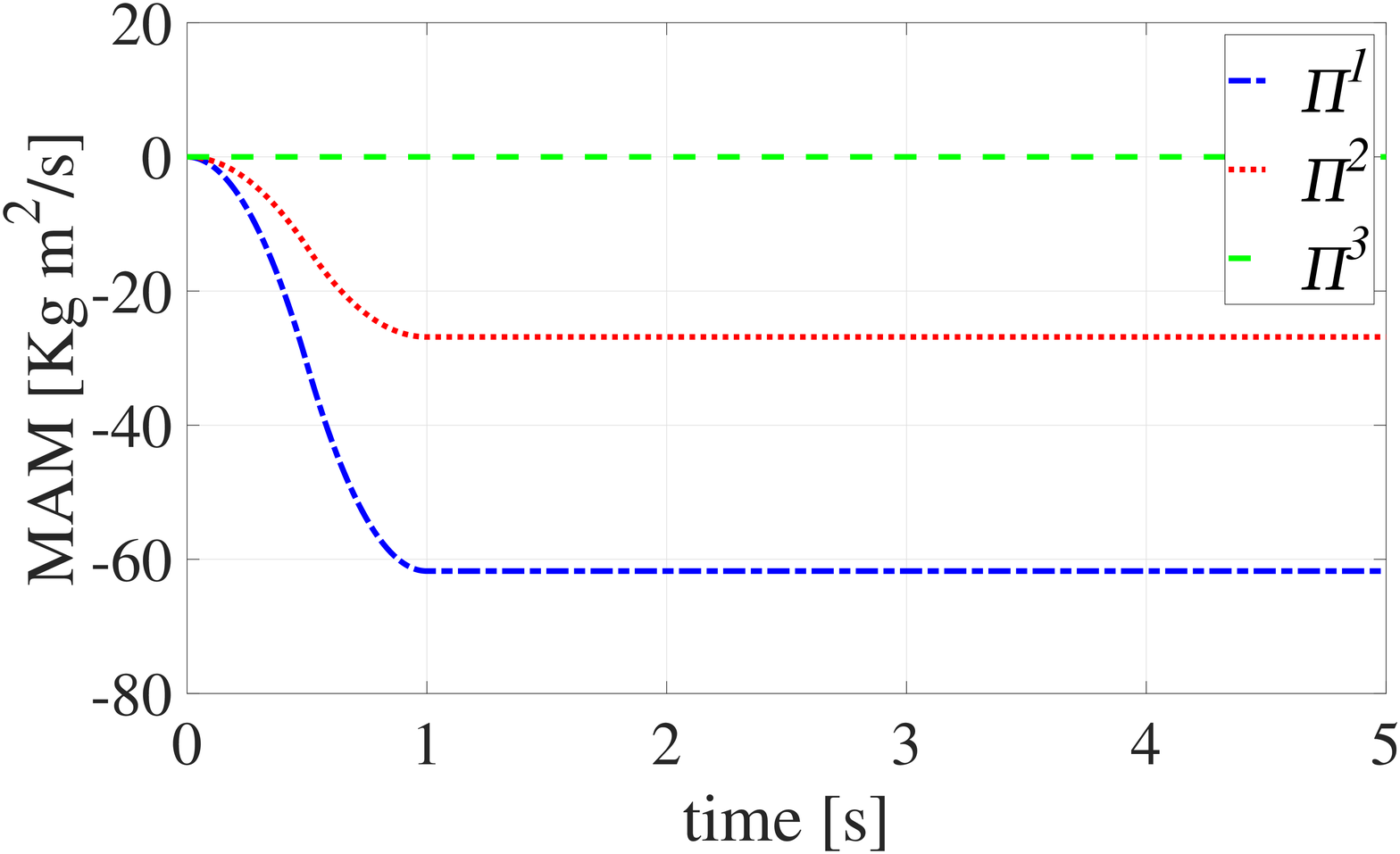}
	\end{tabular}
	\caption{Case 2 - SAM components (left) and MAM components (right).}
	\label{fig-angular-momentum-ml}
\end{figure} 

\begin{figure}[h!]
	\centering{}
	\begin{tabular}{cc}
		\includegraphics[width=0.45\textwidth]{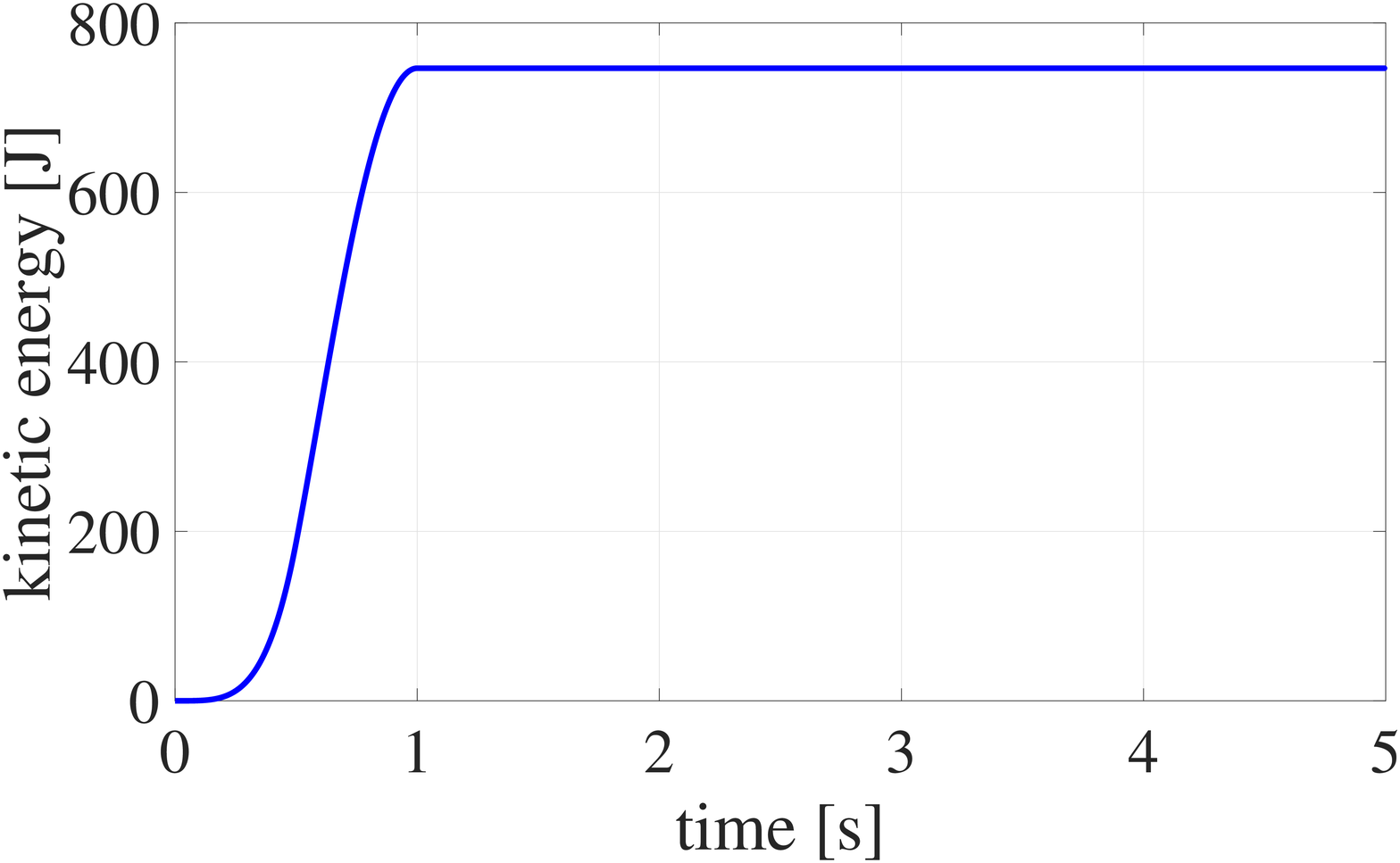} &	\includegraphics[width=0.45\textwidth]{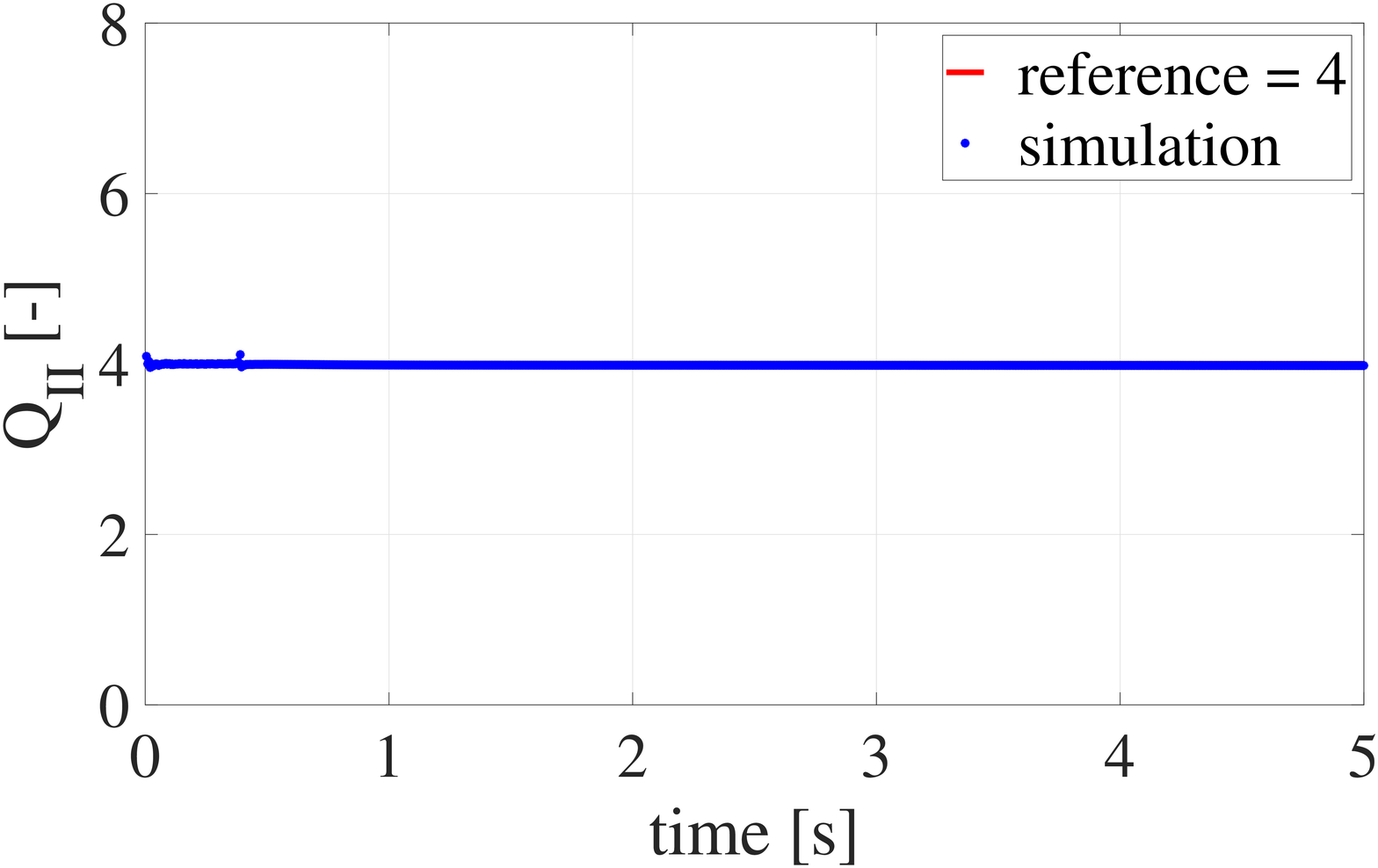}
	\end{tabular}
	\caption{Case 2 - kinetic energy (left) and second quotient of precision (right).}
	\label{fig-energy-precision-ml}
\end{figure}

\begin{figure}[h!]
	\centering{}
	\includegraphics[width=0.5\textwidth]{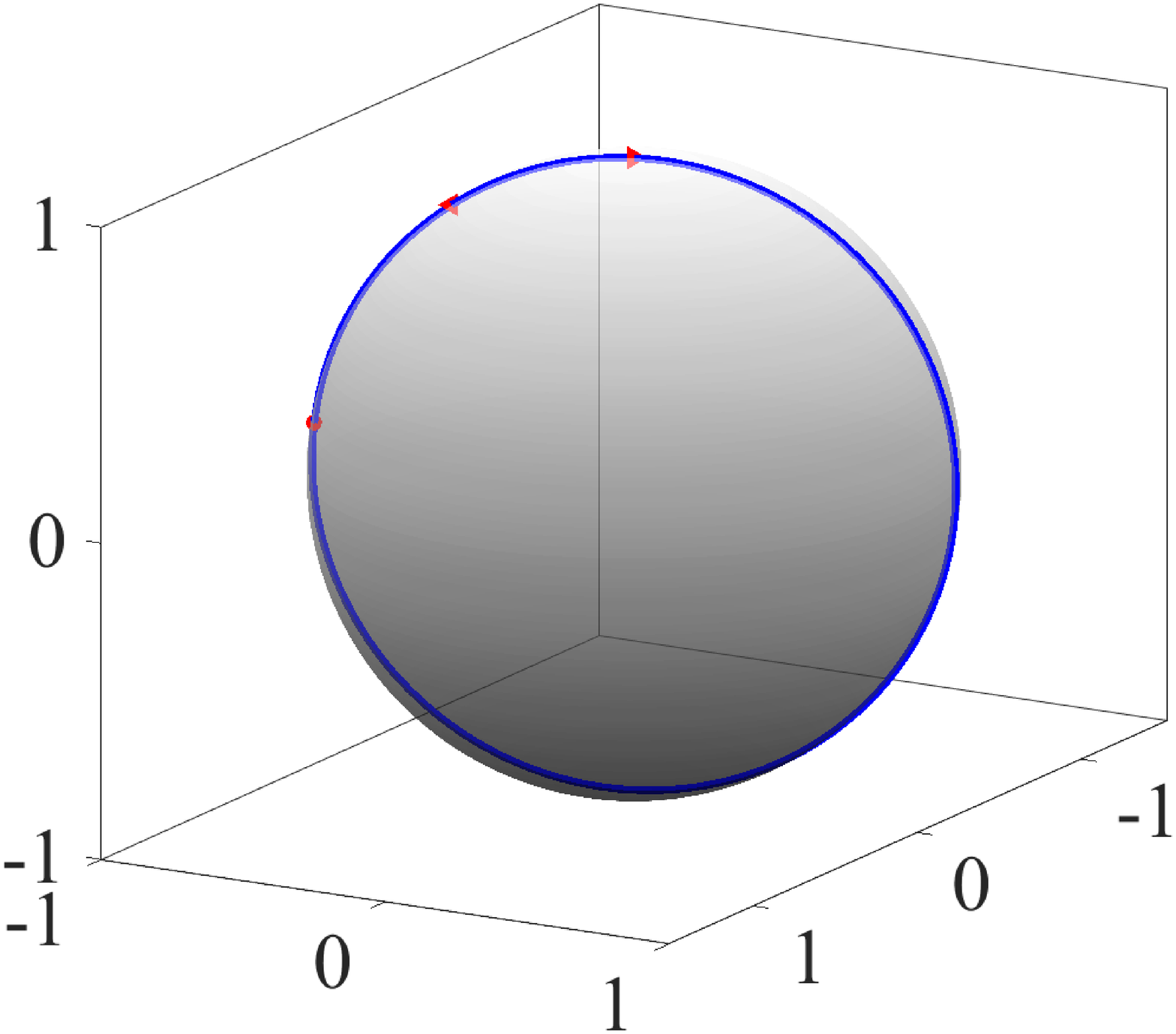}
	\caption{Case 2 - trajectory of $\bm{d}_3$ on $S^2$. }
	\label{fig-trajectory-ml}
\end{figure}

\begin{table}[!ht]
	\centering{}
	\begin{tabular}{|c|c|}
		\hline
		$t$        & $Q_{\mathrm{II}}$ \tabularnewline
		$[$s$]$    &             $[-]$ \tabularnewline
		\hline
		$1.000000$ &        $3.987886$ \tabularnewline
		\hline
		$2.000000$ &        $3.985379$ \tabularnewline
		\hline
		$3.000000$ &        $3.984407$ \tabularnewline
		\hline
		$4.000000$ &        $3.983255$ \tabularnewline
		\hline
		$5.000000$ &        $3.981796$ \tabularnewline		
		\hline		
	\end{tabular}\caption{Case 2 - second quotient of precision.}
	\label{tab-sqp_ml}
\end{table}

\begin{table}[!ht]
	\centering{}
	\begin{tabular}{|c|c|c|c|}
		\hline
		$t$        & $\Pi^1$             & $\Pi^2$             & $K$ \tabularnewline
		$[$s$]$    & $[$Kg$\,$m$^2$/s$]$ & $[$Kg$\,$m$^2$/s$]$ & $[$J$]$ \tabularnewline
		\hline
		$1.000000$ & $-61.749731$        & $-26.845155$        & $746.591699$ \tabularnewline
		\hline
		$2.000000$ & $-61.749731$        & $-26.845155$        & $746.591699$ \tabularnewline
		\hline
		$3.000000$ & $-61.749731$        & $-26.845155$        & $746.591699$ \tabularnewline
		\hline
		$4.000000$ & $-61.749731$        & $-26.845155$        & $746.591699$ \tabularnewline
		\hline
		$5.000000$ & $-61.749731$        & $-26.845155$        & $746.591699$ \tabularnewline		
		\hline		
	\end{tabular}
	\caption{Case 2 - motion invariants - stationary values.}
	\label{tab-invariants_ml}
\end{table}

\subsection{Case 3 - response to nonzero initial conditions and a vanishing load}

For this last case we consider
\begin{equation}
\bm{\Lambda}(0) = \bm{I}\, , \quad \bm{\omega}_\perp(0) = 1.5\bm{d}_1+4.5\bm{d}_2\,\,\mathrm{~rad/s}
\end{equation}
and
\begin{equation}
\bm{m}^{\mathrm{ext}}_\perp(t) =-f(t)(1236.25\bm{d}_1(t)+537.5\bm{d}_2(t))\,\,\mathrm{~ Kgm^2/s^2}
\end{equation}
with $f(t)$ defined as in Eq. \eqref{eq-loads_time_variation}.\\

Figure \ref{fig-angular-momentum-c} presents the time history for the spatial and material
components of the angular momentum. On the left, we can observe that the components of the spatial
angular momentum vary starting from the values corresponding to the initial condition adopted. After
the material load vanishes, the components of the spatial angular momentum oscillate with constant
amplitude and frequency indicating a steady state. On the right, we  observe that the components of
the material angular momentum also vary from the values corresponding to the initial condition
adopted, except the third one that remains always equal to zero. After the material load vanishes
the components of the material angular momenta are identically preserved as in the previous cases.

Figure \ref{fig-energy-precision-c} presents the time history for the kinetic energy and the second
quotient of precision. On the left we can observe that the kinetic energy vary during the first
second, where the applied material load is active. After this vanishes, the kinetic energy is
identically preserved. To the right we see that the second quotient of precision is approximately
$4$, see also Table \ref{tab-sqp_c}.

Figure \ref{fig-trajectory-c} shows the trajectory followed by $\bm{d}_3$. During the first second,
the trajectory does not render a distance minimizing curve on $S^2$. This is due to the combination
of initial condition adopted and the load applied that produces a change in the direction of the
axis of rotation. After the material load vanishes, the trajectory describes a circle of radius
$1$. Table \ref{tab-invariants_c} provides the stationary values for the motion invariants.

\begin{figure}[h!]
	\centering{}
	\begin{tabular}{cc}
		\includegraphics[width=0.45\textwidth]{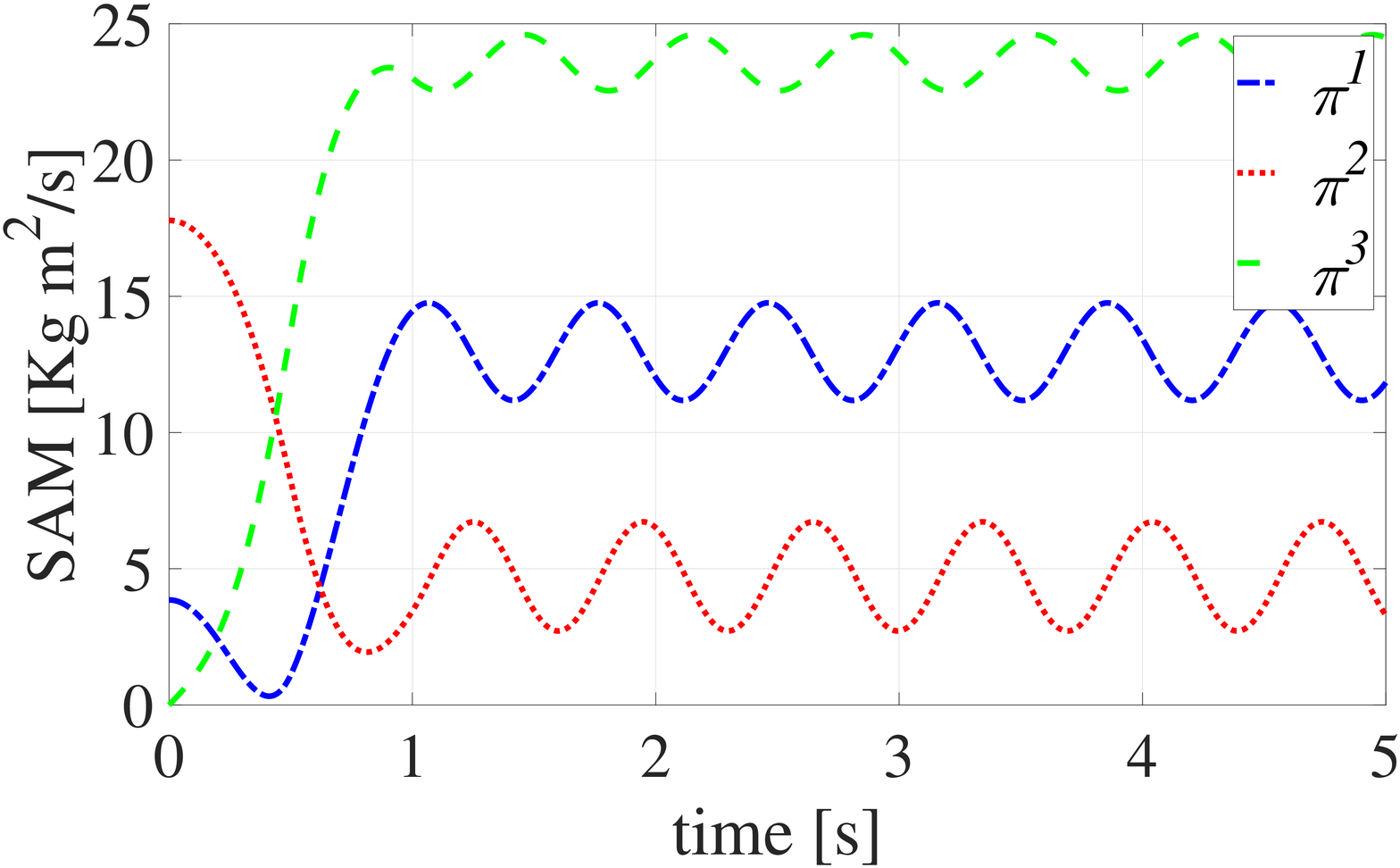} &	\includegraphics[width=0.45\textwidth]{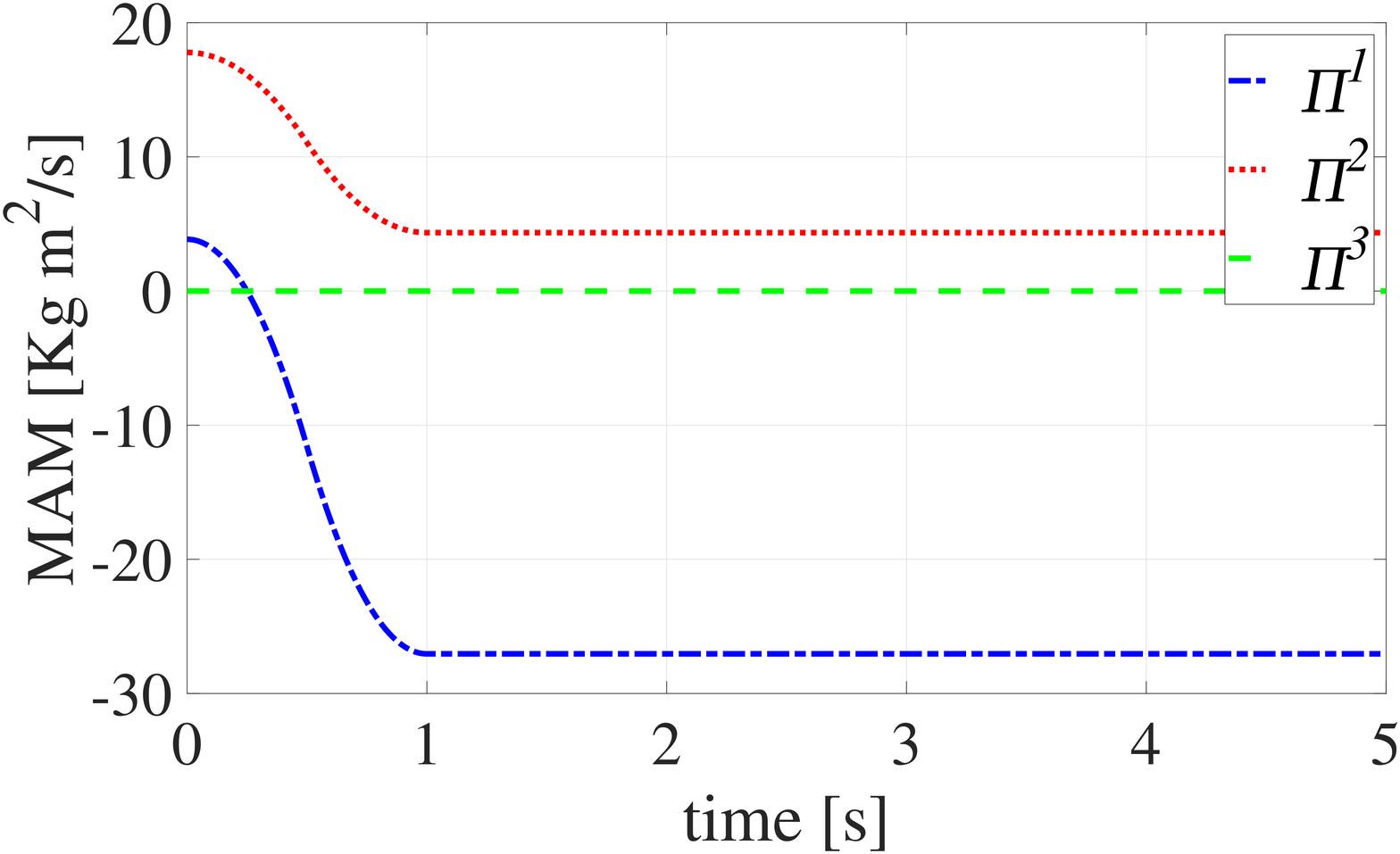}
	\end{tabular}
	\caption{Case 3 - SAM components (left) and MAM components (right).}
	\label{fig-angular-momentum-c}
\end{figure} 

\begin{figure}[h!]
	\centering{}
	\begin{tabular}{cc}
		\includegraphics[width=0.45\textwidth]{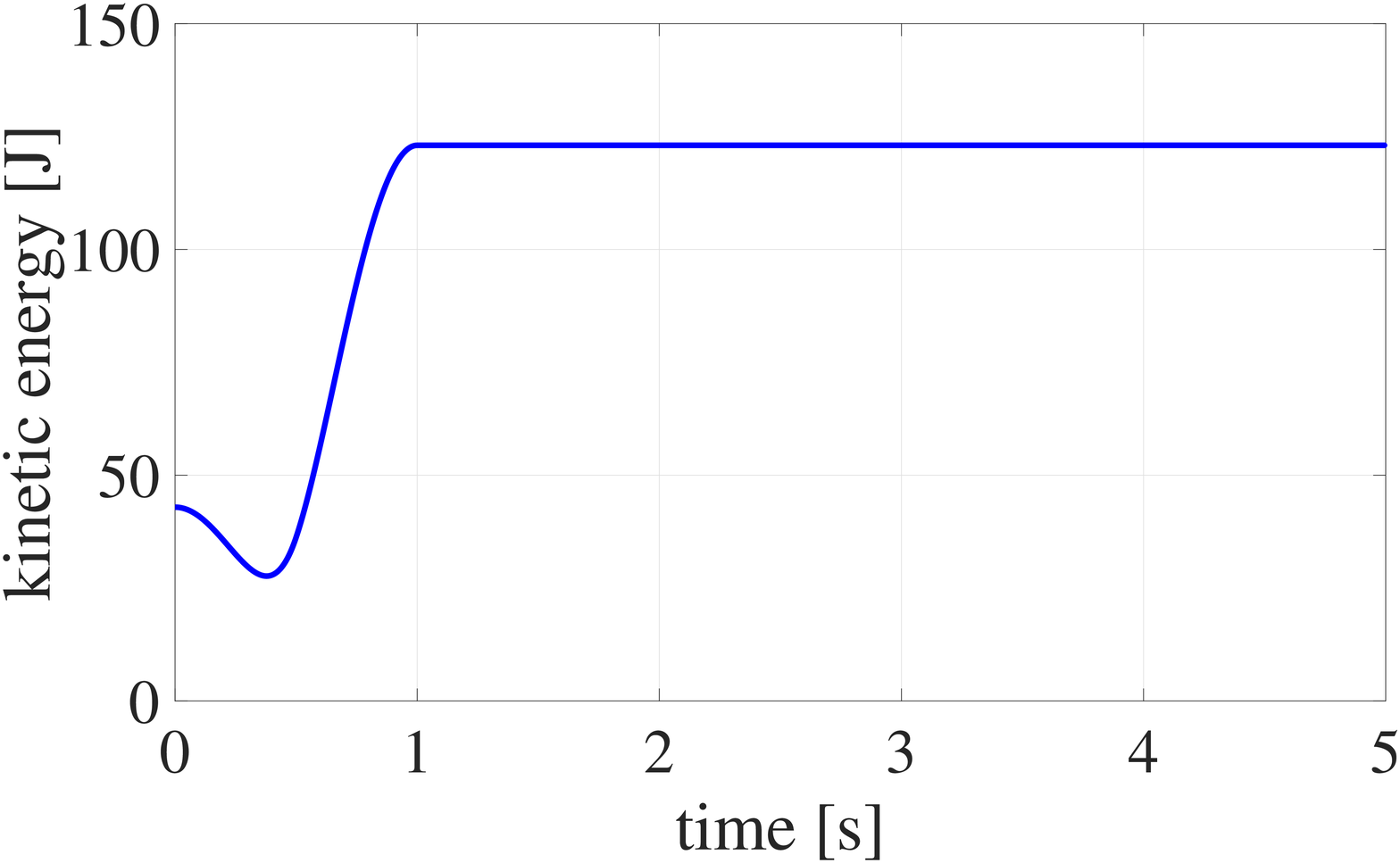} &	\includegraphics[width=0.45\textwidth]{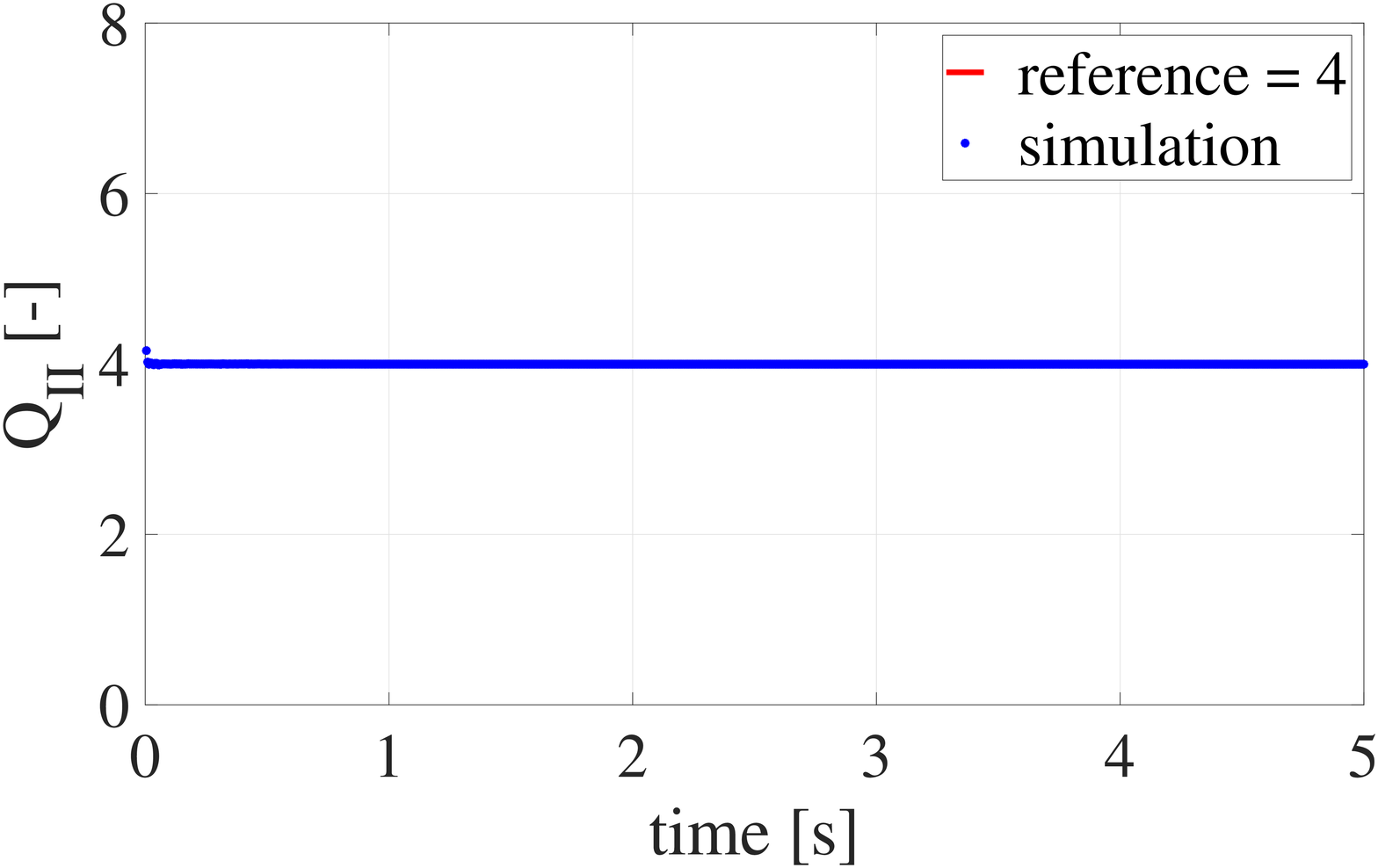}
	\end{tabular}
	\caption{Case 3 - kinetic energy (left) and second quotient of precision (right).}
	\label{fig-energy-precision-c}
\end{figure}

\begin{figure}[h!]
	\centering{}
	\includegraphics[width=0.5\textwidth]{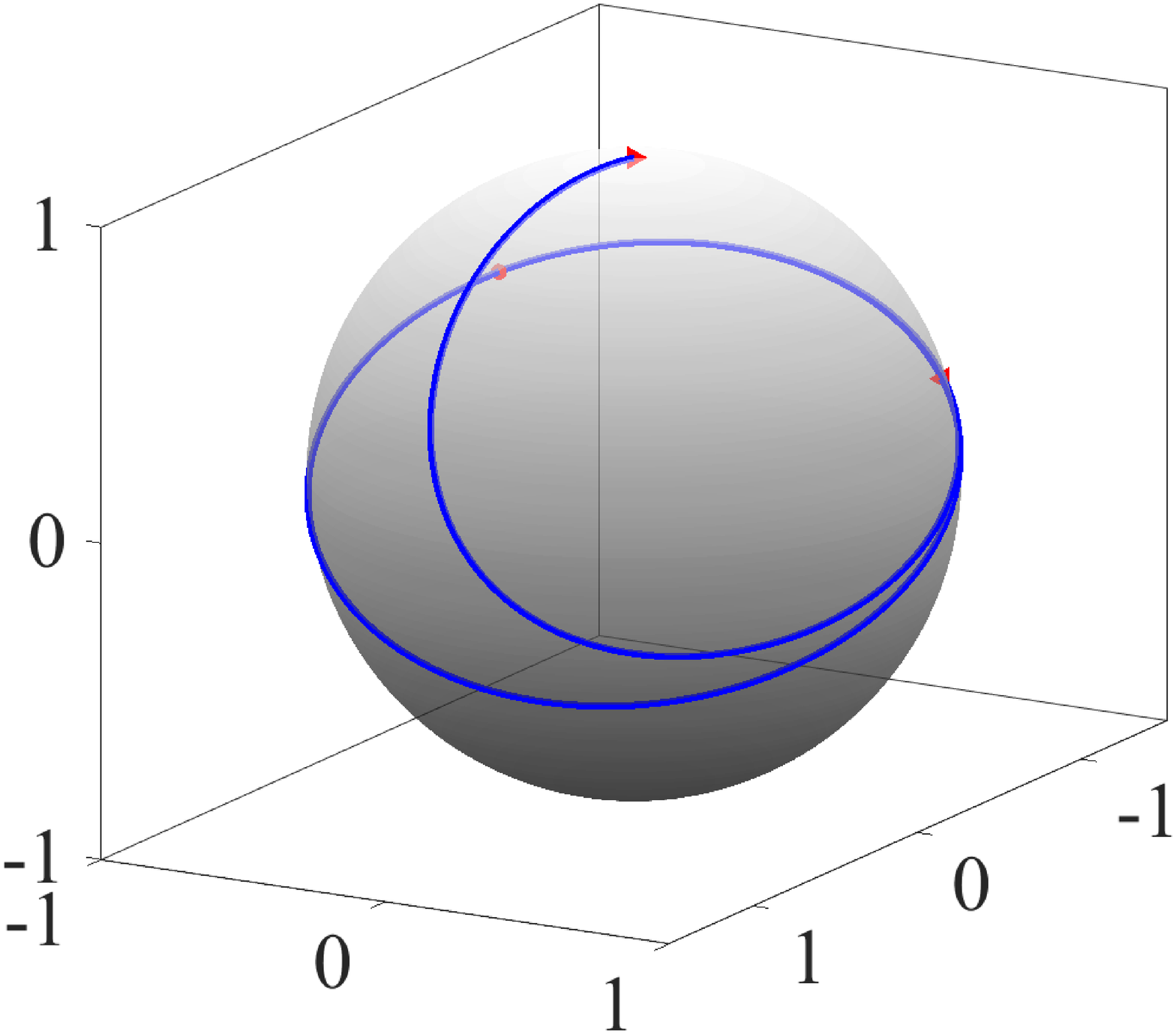}
	\caption{Case 3 - trajectory of $\bm{d}_3$ on $S^2$. }
	\label{fig-trajectory-c}
\end{figure}

\begin{table}[!ht]
	\centering{}
	\begin{tabular}{|c|c|}
		\hline
		$t$        & $Q_{\mathrm{II}}$ \tabularnewline
		$[$s$]$    &             $[-]$ \tabularnewline
		\hline
		$1.000000$ &        $3.998023$ \tabularnewline
		\hline
		$2.000000$ &        $3.997424$ \tabularnewline
		\hline
		$3.000000$ &        $3.997327$ \tabularnewline
		\hline
		$4.000000$ &        $3.997289$ \tabularnewline
		\hline
		$5.000000$ &        $3.997247$ \tabularnewline		
		\hline		
	\end{tabular}\caption{Case 3 - second quotient of precision.}
	\label{tab-sqp_c}
\end{table}

\begin{table}[!ht]
	\centering{}
	\begin{tabular}{|c|c|c|c|}
		\hline
		$t$        & $\Pi^1$             & $\Pi^2$             & $K$ \tabularnewline
		$[$s$]$    & $[$Kg$\,$m$^2$/s$]$ & $[$Kg$\,$m$^2$/s$]$ & $[$J$]$ \tabularnewline
		\hline
		$1.000000$ & $-27.042702$        & $4.351472$        & $123.059926$ \tabularnewline
		\hline
		$2.000000$ & $-27.042702$        & $4.351472$        & $123.059926$ \tabularnewline
		\hline
		$3.000000$ & $-27.042702$        & $4.351472$        & $123.059926$ \tabularnewline
		\hline
		$4.000000$ & $-27.042702$        & $4.351472$        & $123.059926$ \tabularnewline
		\hline
		$5.000000$ & $-27.042702$        & $4.351472$        & $123.059926$ \tabularnewline		
		\hline		
	\end{tabular}
	\caption{Case 3 - motion invariants - stationary values.}
	\label{tab-invariants_c}
\end{table}


\section{Summary}
\label{sec-summary}
This article describes the governing equations of the rotating rigid body
in a nonholonomic context, and discusses their relation with other,
well-known, equivalent models based on rotations and orthonormal vectors.
The equations obtained are non-variational and possess first
invariants of motion. Some of them, \ie the nonholonomic momenta (first and second components of the material angular momentum), are neither evident from the standard descriptions nor intuitive. To the best of our knowledge, there is no work in the literature that reports similar observations and thus, it represents a main innovation of the current work.

Complementing the rigorous mathematical analyis done for the proposed model, an implicit, second order accurate, energy and momentum conserving algorithm is presented, which discretizes in time the rigid body, nonholomic equations. Such a time integration scheme preserves exactly the energy and nonholonomic momenta and thus, this represents also a main innovation of the current work. Finally, simple examples, which make use of all elements of the approach proposed, are provided and confirm the excellent conservation properties and second order accuracy of the new scheme.


\appendix

\section{Second quotient of precision}

To investigate the correctness
of the integration scheme, we can employ the second quotient of precision \cite{Kreiss2014}. Any numerical solution of an initial
value problem can be written as
\begin{equation}
\bm{\xi}(t,h,k)=
\bm{\xi}(t)+
\left(\frac{h}{k}\right)\bm{\psi}_{1}(t)+
\left(\frac{h}{k}\right)^{2}\bm{\psi}_{2}(t)+
\ldots+
\left(\frac{h}{k}\right)^{n}\bm{\psi}_{n}(t)+
\mathcal{O}(h^{n+1})
\end{equation}
with $\bm{\xi}(t)$ representing the exact solution of the initial
value problem under consideration and $\bm{\psi}_{i}$ for $i=1,\ldots,n$ representing
smooth functions that only depend on the time parameter $t$. $h$ stands for the time step and $k$ is a positive
integer number that enables defining finer solutions computed from the
original resolution.

Let us introduce the second quotient of precision given by
\begin{equation}
Q_{\textrm{II}}(t)=\frac{\left\| \bm{\xi}(t,h,1)-\bm{\xi}(t,h,2)\right\| }{\left\| \bm{\xi}(t,h,2)-\bm{\xi}(t,h,4)\right\| }\,.
\end{equation}
For the numerator, we have that
\begin{equation}
\begin{aligned}
\left\| \bm{\xi}(t,h,1)-\bm{\xi}(t,h,2)\right\| 
& = \left\| \left(\frac{h}{1}\right)^{n}\bm{\psi}_{n}(t)-\left(\frac{h}{2}\right)^{n}\bm{\psi}_{n}(t)+\mathcal{O}(h^{n+1})\right\| \\
& = \left(\frac{2^{n}-1}{2^{n}}\right)h^{n}\left\| \bm{\psi}_{n}(t)\right\| +\mathcal{O}(h^{n+1})
\end{aligned}
\end{equation}
meanwhile, for the denominator, we have that
\begin{equation}
\begin{aligned}
\left\| \bm{\xi}(t,h,2)-\bm{\xi}(t,h,4)\right\|
& = \left\| \left(\frac{h}{2}\right)^{n}\bm{\psi}_{n}(t)-\left(\frac{h}{4}\right)^{n}\bm{\psi}_{n}(t)+\mathcal{O}(h^{n+1})\right\| \\
& = \left(\frac{2^{n}-1}{4^{n}}\right)h^{n}\left\| \bm{\psi}_{n}(t)\right\| +\mathcal{O}(h^{n+1})\,.
\end{aligned}
\end{equation}
If we have time steps that are sufficiently small, it can be showed that the second quotient of precision is very close to take the value $2^{n}$,
\ie
\begin{equation}
Q_{\textrm{II}}(t)=
\frac{\left(\frac{2^{n}-1}{2^{n}}\right)h^{n}\left\| \bm{\psi}_{n}(t)\right\| +\mathcal{O}(h^{n+1})}{\left(\frac{2^{n}-1}{4^{n}}\right)h^{n}\left\| \bm{\psi}_{n}(t)\right\| +\mathcal{O}(h^{n+1})}=
2^{n}+\mathcal{O}(h^{n+1})
\approx
2^{n}\,.
\label{eq:second_quotient}
\end{equation}

In this work, we consider an integration algorithm whose accuracy is warranted to be of second order and thus, $Q_{\textrm{II}}(t)\approx4$. Be aware that for the computation of
precision quotients, $h$ needs to be chosen small enough. Such a selection strongly depends on the case considered.
Moreover, when $\|\bm{\psi}_{n}(t)\|$
is very small, the test may fail even if correctly implemented. Therefore, we recommend to play with initial conditions and time resolutions for achieving correct pictures.

\bibliographystyle{ieeetr}
\bibliography{bib}

\end{document}